\documentclass[10pt,a4paper]{article}
\usepackage[top=60pt,bottom=60pt,left=82pt,right=82pt]{geometry}

\usepackage{subcaption}
\usepackage[utf8]{inputenc}
\usepackage[linesnumbered, ruled, vlined]{algorithm2e}
\usepackage{amssymb}
\usepackage{amsmath}
\usepackage{nccmath}
\usepackage{amsthm}
\usepackage{thmtools}
\usepackage{thm-restate}
\usepackage{hyperref}
\usepackage[capitalise]{cleveref}
\usepackage[normalem]{ulem}
\usepackage{mathtools}
\usepackage{empheq}
\usepackage{fontawesome}
\usepackage{graphicx}
\usepackage{bm}

\newcommand{\setD}{{\mathcal D}}
\newcommand{\setS}{{\mathcal S}}  
\newcommand{\setSp}{{\setS^+}}
\newcommand{\setQ}{\{ q_\serv \}_{\serv\in\setS}}  
\newcommand{\setDu}{\{ \rate_\serv \}_{\serv\in\setS}} 
\newcommand{\setO}{{\mathcal O}} 
 
\newcommand{\serv}{s}%notation for servers
\newcommand{\servT}{r}%notation for servers
\newcommand{\servU}{u}%notation for servers
\newcommand{\disp}{d}%dispatchers
\newcommand{\receives}{\bar{a}}%how many jobs a server receives

\newcommand\set[1]{\{#1\}}
\newcommand\T[1]{\noindent\textbf{#1}}

\newcommand{\widthToFill}{\rate_{tot}}
\newcommand{\filled}{l} %{\textit{Filled}~}
\newcommand{\curr}{r}
\newcommand{\val}{\textit{val}}
\newcommand{\rate}{\mu}%{w}
\newcommand{\iba}{\textsc{iba}}%{w}
\newcommand{\iwl}{\textsc{iwl}}%{w}
\newcommand{\const}{z}
\newcommand{\diff}{\text{diff}}
\newcommand{\SCD}{\textsc{StochasticallyCoordinatedDispatching}}

\newcommand{\brac}[1]{\left\{ #1 \right\}}

\newcommand{\Natz}{\mathbb{N}}

\newcommand{\wl}{\iwl}%{WL}

\newcommand{\TT}[1]{\noindent\textbf{#1}}

\newcommand{\bp}[1]{\Big(#1\Big)}

\providecommand{\ie}{\emph{i.e.,} }
\providecommand{\eg}{\emph{e.g.,} }

\DeclareMathOperator{\E}{\mathbb{E}}

\newtheorem{corollary}{Corollary}
\newtheorem{lemma}{Lemma}
\newtheorem{definition}{Definition}

\title{Stochastic Coordination in Heterogeneous\\ Load Balancing Systems}

\author{ 
    Guy Goren\\
	\texttt{\small Technion}
	\and 
	Shay Vargaftik\\
	\texttt{\small VMware Research}
	\and 
	Yoram Moses\\
	\texttt{\small Technion}
}

\begin{document}

\date{}
\maketitle

% you can include author information in the source, but `anonymous' option will hide it

% \author{First Author}
% \email{first.author@email.com}
% \affiliation{%
%   \institution{Affiliation of first author}
% }

% \author{Second Author}
% \email{second.author@email.com}
% \affiliation{%
%   \institution{Affiliation of second author}
% }

% \author{Second Author}
% \email{second.author@email.com}
% \affiliation{%
%   \institution{Affiliation of second author}
% }

% \author{Guy Goren}
% \email{sgoren@campus.technion.ac.il}
% \affiliation{%
%   \institution{Technion, Israel}
% }

% \author{Shay Vargaftik}
% \email{shayv@vmware.com}
% \affiliation{%
%   \institution{VMware Research}
% }

% \author{Yoram Moses}
% \email{moses@ee.technion.ac.il}
% \affiliation{%
%   \institution{Technion, Israel}
% }

\begin{abstract}
 Current-day data centers and high-volume cloud services employ a broad set of heterogeneous servers. In such settings, client requests typically arrive at multiple entry points, and dispatching them to servers is an urgent distributed systems problem. 
 This paper presents an efficient solution to the load balancing problem in such systems that improves on and overcomes problems of previous solutions. The load balancing problem is formulated as a stochastic optimization problem, and an efficient algorithmic solution is obtained based on a subtle mathematical analysis of the problem. Finally, extensive evaluation of the solution on simulated data shows that it outperforms previous solutions. Moreover, the resulting dispatching policy can be computed very efficiently, making the solution practically viable. 
\end{abstract}

% keywords, ACM classification and conference information can be omitted for submission

% \begin{itemize}
%     \item Abstract 0.5
%     \item Intro 2 
%     \item Related \& Model 2
%     \item Ideal 1.5
%     \item Optimization problem 1.5
%     \item Algorithmic  solution 2
%     \item Putting it all 1
%     \item Evaluation 3
%     \item Discussion 1
% \end{itemize}

%
\section{Introduction}

%%%%%%%%%%%%%%%%%%%%%%%%%%%%%%%%%%
%%% background & significance %%%% 
%%%%%%%%%%%%%%%%%%%%%%%%%%%%%%%%%%

Load balancing in modern computer clusters is a challenging task. Unlike in the traditional parallel server model where all client requests arrive through a single centralized entry point, today's cluster designs are distributed~\cite{barbette2020high,gandhi2014duet,eisenbud2016maglev,prekas2017zygos}. In particular, they involve many dispatchers that serve as entry points to client requests and distribute these requests among a 
multitude of servers.
%
%%%%%%%%%%%%%%%%%%%%%
%%% motivation ##%%%% 
%%%%%%%%%%%%%%%%%%%%%
%
The dispatchers' goal is to distribute the client requests in a balanced manner so that no server is overloaded or underutilized. This is particularly challenging due to two system design attributes: (1) the dispatchers must take decisions immediately upon arrival of requests, and independently from each other. This requirement is critical to adhere to the high rate of incoming client requests and to the extremely low required response times~\cite{dean2013tail,nishtala2017hipster}. Indeed, even a small sub-second addition to response time in dynamic content websites can lead to a persistent loss of users and revenue~\cite{lu2011join,schurman2009user}; (2) today's systems are heterogeneous with different servers containing different generations of CPUs, various types of acceleration devices such as GPUs, FPGAs, and ASICs, with varying  processing speeds~\cite{govindan2016evolve,huang2016programming,delimitrou2013paragon,mars2011heterogeneity,duato2010rcuda}. 
This paper presents a new load balancing solution for distributed dispatchers in heterogeneous systems.

Most previous works on distributed load balancing focused either on homogeneous systems or on systems in which dispatchers have limited information about server queue-lengths. 
In contrast, in today's heterogeneous systems the dispatchers (e.g., high-performance production L7 load balancers such as HAProxy~\cite{haproxy_gen} and NGINX~\cite{ngynx_gen}) \mbox{typically have access to abundant queue-length information.} 

%%%%%%%%%%%%%%%%%%%%%%%%%%%%%%%%%%%%%%%%%%%%%%%%%%%%%%%%%%%%%%%
%%% herding and stochastic coordination%%%%%%%%%%%%%%%%%%%%%%%% 
%%%%%%%%%%%%%%%%%%%%%%%%%%%%%%%%%%%%%%%%%%%%%%%%%%%%%%%%%%%%%%%

A series of recent works have shown that popular ``{\it join-the-shortest-queue} (JSQ)''-based load-balancing policies behave poorly when the dispatchers' information is highly correlated~\cite{mitzenmacher2000useful,vargaftik2020lsq,DISC2020,zhou2020asymptotically}. 
In particular, such policies suffer from so-called ``{\it herd behavior},'' in which  different dispatchers concurrently recognize the same set of less loaded servers and forward all their incoming requests to this set. 
The queue-lengths of the servers in this small set then grow rapidly, causing excessive processing delays, increased response times, and even dropped jobs. 
In severe cases, herding may even result in servers stalling and crashing.

Herding has been dubbed the ``{\it finger of death}'' in~\cite{youtube_lec}.
In order to avoid it, both in recent theory works~\cite{vargaftik2020lsq,zhou2020asymptotically} and in modern state-of-the-art production deployments~\cite{haproxy_po2,ngynx_po2}, researchers introduced load balancing policies that break symmetry among the dispatchers via a random subsampling of queue-length information.
Even large cloud service companies such as Netflix report making use of such limited queue-size information policies to avoid suffering from detrimental herd behavior effects~\cite{youtube_lec,netflixEdge}.  
They choose to do so despite the potential for degraded resource utilization and worse response times.
%\guy{maybe repeat in the last sentence the trade-off:\\
 %That is, they are willing to pay in average service speed in order to avoid these effects.}\shay{I think its redundant... we don't want to overstate thinks regarding production stuff...}

 The better-information/worse-performance paradox manifested by herding was recently addressed for load balancing in homogeneous systems (in which servers are all equally powerful) by \cite{DISC2020}.
 They demonstrated that the tradeoff between more accurate information and herding is not inherent. 
 Rather, herding stems from the fact that currently employed dispatching techniques do not account for the fact that multiple dispatchers operate concurrently in the system. 
 Based on this observation, they suggested a solution in which the dispatchers' policies address the presence and concurrent operation of multiple dispatchers by employing \emph{stochastic coordination}~(Section ~\ref{sec:Related work}). 
 While their solution provides superior performance in homogeneous systems, it performs poorly in heterogeneous systems, since it  does not account for variation in server service rates. 
 
%%%%%%%%%%%%%%%%%%%%%
%%% our %%%%%%%%%%%%% 
%%%%%%%%%%%%%%%%%%%%%

Our work generalizes the approach of~\cite{DISC2020} and obtains an efficient load-balancing policy based on stochastic coordination for the more challenging heterogeneous case. We initially expected that mild adjustments to the scheme of~\cite{DISC2020} for homogeneous systems would yield similarly effective policies for the heterogeneous case. That turned out not to be the case. The generalization required overcoming nontrivial hurdles at the level of mathematical analysis, of the algorithmic treatment, and at the conceptual level. The solution that we obtained has several features that the policy presented in~\cite{DISC2020} does not have. In both cases, a randomized load-balancing policy is designed based on the solution of a stochastic optimization problem. In the homogeneous case, \cite{DISC2020} have shown that  this solution can be formulated in terms of a closed-form formula that can be efficiently computed in real-time by the dispatchers. In contrast, explicitly computing the solution obtained for the heterogeneous case requires exponential time. Even obtaining an approximate solution using standard optimization techniques requires cubic time, which is not feasible for dispatchers to perform. Using a novel analysis of the optimization problem in the heterogeneous case, we show that the search for an optimal solution can be made in an extremely efficient manner, which results in a practically feasible algorithm. We thus obtain a novel policy, we term \emph{stochastically coordinated dispatching} ($SCD$) that greatly improves over existing load-balancing policies for heterogeneous systems, and is essentially as easy to compute as the simplest policies are.

To evaluate the performance of $SCD$ we implemented it as well as 10 other algorithms (including both traditional techniques and recent state-of-the-art ones) in C++.
Extensive evaluation results indicate that $SCD$ consistently outperforms all tested techniques over different systems, heterogeneity levels and metrics.
For example, at high loads, $SCD$ improves the 99th percentile delay of client requests by more than a factor of $2$ in comparison to the second-best policy, and by more than an order of magnitude compared to the heterogeneity-oblivious solution in \cite{DISC2020}. In terms of computational running time, $SCD$ is competitive with currently employed techniques (e.g., $JSQ$). For example, even in a system with 100 servers, $SCD$ requires only a few microseconds on a single CPU core to make dispatching decisions.  
Finally, our results are reproducible. Our $SCD$ implementation is available on GitHub~\cite{scd_git_code}.

%%%%%%%%%%%%%%%%%%%%%
%%% related work %%%% 
%%%%%%%%%%%%%%%%%%%%%

\subsection{Related work}\label{sec:Related work}

%%% policies that were designed for a single dispatcher - classic
In a traditional computer cluster with a centralized design and a single dispatcher that takes all the decisions, a centralized algorithm such as $JSQ$, that assigns each arriving request to the currently shortest queue, offers favorable performance and strong theoretical guarantees~\cite{weber1978optimal,winston1977optimality,eryilmaz2012asymptotically}. 
However, in a distributed design, where each dispatcher independently follows $JSQ$, the aforementioned herding phenomenon occurs~~\cite{youtube_lec,netflixEdge,vargaftik2020lsq,zhou2020asymptotically}.
As a result, both researchers and system designers often resort to traditional techniques that were originally designed either for centralized systems with a single dispatcher or for systems based on limited queue state information.
%
%%% policies that were designed for a single dispatcher - reduced state
For example, in the {\it power-of-$d$-choices} policy (denoted $JSQ(d)$)~\cite{luczak2006maximum,vvedenskaya1996queueing,mitzenmacher2001power}, when requests arrive, a dispatcher samples $d$ servers uniformly at random and employs $JSQ$ considering only the $d$ sampled servers. 
This policy alleviates the herding phenomenon for sufficiently low $d$ values since it is only with a low probability that different dispatchers will sample the same good server(s) at a given point in time. 
However, low $d$ values often come at the price of longer response times and low resource utilization~\cite{foss1998stability}, whereas herding does occur for higher $d$ values. 
Moreover, in heterogeneous systems, $JSQ(d)$ may even result in \emph{instability}\footnote{A queue is unstable when its size continues to grow in an unbounded manner. 
A load balancing system is unstable if at least a single queue in the system is unstable. 
Instability in heterogeneous load balancing systems usually occurs when the faster queues are constantly idling because they do not receive enough requests, whereas slower servers receive too many requests, and their queue continues to grow.} for any $d$ value strictly below the number of servers.

%A related dispatching policy is termed {\it power-of-memory} (denoted by $JSQ(d,m)$). In $JSQ(d,m)$, in addition to $d$ randomly chosen servers, the dispatcher also samples the $m\le d$ shortest queues to whom it sent jobs in the latest round. Each job is then routed to the shortest among these $d + m$ queues. Although $JSQ(d,m)$ is known to have strong theoretical guarantees for a single dispatcher~ \cite{shah2002use,mitzenmacher2002load}, in a distributed system its behaviour with respect to $d$ and $m$ suffers from similar drawbacks as $JSQ(d)$~\cite{vargaftik2020lsq}.

%%% adapted to heterogeneity, but consider only a single dispatcher
Recently, to account for server heterogeneity, {\it shortest-expected-delay} ($SED$) policies were proposed. 
These policies operate similarly to $JSQ$ and $JSQ(d)$ but, instead of ranking servers according to their queue-lengths, servers are compared according to their {\it normalized} queue-lengths, i.e., their queue length divided by their processing capacity. 
This way, a job is sent not to the server with the shortest queue but to the server with the shortest expected wait time. 
Indeed these policies were shown to outperform their heterogeneity-unaware $JSQ$-based counterparts in heterogeneous systems (for homogeneous systems these policies coincide)~\cite{gardner2021scalable,jaleelgeneral,gardner2019smart,selen2016steady}. 
However, in a setting with multiple, distributed dispatchers, the $SED$ policies suffer from the same herding phenomenons.

%%% policies that were designed for homogeneous distributed systems
The first dispatching policy that was designed specifically for the multi-dispatcher case, called {\it join-the-idle-queue} ($JIQ$)~\cite{lu2011join,mitzenmacher2016analyzing,stolyar2017pull,stolyar2015pull,van2017load}, was originally introduced by Microsoft~\cite{jiq_microsoft}.
In $JIQ$, a dispatcher sends requests only to idle servers. 
If there are no idle servers, the requests are forwarded to randomly chosen ones. 
$JIQ$ significantly outperforms $JSQ(d)$ when the system operates at low loads.
However, its performance quickly deteriorates when the load increases and the dispatching approaches a random one. 
In fact, as is the case with $JSQ(d)$, the $JIQ$ policy may exhibit instability in the presence of high loads~\cite{zhou2017designing,atar2020persistent}.  
A recent work~\cite{gardner2021scalable} considered an improvement to $JIQ$ that accounts for server heterogeneity by adapting the server sampling probabilities to account for their processing rate. 
While this policy restores stability when the load is high, it is significantly outperformed by policies such as $SED$ and even $JSQ$ when queue-length information is available~\cite{zhou2020asymptotically}.

%%% designed for heterogeneous distributed systems
Recent state-of-the-art techniques include the local-shortest-queue ($LSQ$)~\cite{vargaftik2020lsq} and the {\it local-estimation-driven} ($LED$) \cite{zhou2020asymptotically} policies. They address the limitations of both $JSQ(d)$ and $JIQ$ by maintaining a local array of server queue-lengths at each dispatcher. Each dispatcher updates its local array by randomly choosing servers and querying them for their queue-length. 
Consequently, the dispatchers have different views of the server queue-lengths. 
However, the performance of $LSQ$ and $LED$ depends on the dispatchers' local arrays being weakly correlated. When the dispatchers' queue length information is even partially correlated, both $LSQ$ and~$LED$ incur herding~\cite{DISC2020}. % \yoram{Do we need the next sentence?} Moreover, both policies converge to $JSQ$ when timely system information is available, which is the case in many modern clusters.  

%%% TWF
As discussed in the Introduction, the seemingly paradoxical herding behavior was recently addressed in the case of homogeneous systems in~\cite{DISC2020}. Specifically, they introduced the {\it  tidal-water-filling} policy (denoted $TWF$). 
In $TWF$, the dispatching policy of a dispatcher is defined by the probabilities at which it sends each arriving request to each server. The main idea for utilizing accurate server queue-lengths information without incurring herding relies on stochastic coordination of the dispatchers --- i.e., setting these probabilities such that all the dispatchers' decisions combined result in a balanced state. 
Nevertheless, $TWF$ does not account for server service rates. Consequently, its performance significantly degrades in a heterogeneous system, resulting in reduced resource utilization and excessively long response times.
This hinders its applicability to modern computer clusters. 

%%%%%%%%%%%%%%%%%%%%%%%%%%%%%%%%%%%%%%%%%%%%%%%%%%%%%
%%% lessons learned from previous submission %%%%%%%%
%%%%%%%%%%%%%%%%%%%%%%%%%%%%%%%%%%%%%%%%%%%%%%%%%%%%%

A related line of work dealing with load balancing challenges in distributed systems is based on the balls-into-bins model \cite{adler1998parallel}, including extensions to dynamic or heterogeneous settings (\eg \cite{azar1994balanced,berenbrink2014balls}). In the balls-into-bins model, it is commonly possible to obtain more precise theoretical guarantees. Indeed, common approaches include regret minimization (\eg \cite{kleinberg2011load}) and adaptive techniques (\eg \cite{lenzen2011tight}). Nevertheless, this model is not aligned with our model (\eg we consider multiple dispatchers, stochastic arrivals at each dispatcher and stochastic departures at each server). As a result, their analysis does not apply in our model, and vice versa.

A seminal work in~\cite{karp1996efficient} deals with parallel accesses of CPUs to memory regions, trying to minimize collisions. Their solutions rely on hash functions to prevent the memory accesses from becoming the bottleneck for system performance. In our setting, using hash functions resembles random allocation which is known to be sub-optimal.    

Another related line of work concerns static load balancing,  where the goal is to assign jobs to servers in a distributed manner to optimize some balance metric and with a minimal number of communication rounds among the participants (usually in the CONGEST or LOCAL model)~\cite{assadi2020improved,czygrinow2012distributed,halldorsson2018distributed}. This model is not aligned with ours since we consider dynamic systems with the demand for immediate and independent decision making among the dispatchers to sustain the high incoming rate of client requests and the demand for low latency. In particular, in our model, the dispatchers do not interact.\footnote{In practice, any communication among the distributed dispatchers introduces additional processing and, more importantly, possibly unpredictable network delay. Therefore, the dispatchers' high rate of incoming jobs makes interaction among them highly undesirable and even not feasible, especially when the dispatchers are not co-located (i.e., not on the same machine). As a result, assuming that dispatchers do not interact is standard practice in our L7 cluster load balancing model~\cite{zhou2020asymptotically,vargaftik2020lsq,wang2018distributed,stolyar2017pull,stolyar2015pull,lu2011join,haproxy_po2,ngynx_po2,DISC2020}.}

\section{Model}\label{sec:Model}

We consider a system with a set $\setS$ of $n$ servers and a set $\setD$ of $m$ dispatchers. The system operates over discrete and synchronous rounds $t \in \Natz $. 
Each server $\serv\in\setS$ has its own FIFO queue\footnote{FIFO stands for first-in-first-out. Namely, client requests at each server are processed in the order at which they arrive at the queue. The order among client requests that arrive at the same time is arbitrary.} of pending client requests. 
We denote by $q_\serv(t)$ the number of client requests at server's $\serv$ queue at the beginning of round $t$. We assume that the values of $q_\serv(t)$ for all servers~$\serv$ are available in round~$t$ to all dispatchers, for all $t\in\mathbb{N}$.  Each round consists of three phases:

\begin{enumerate}

  \item \T{Arrivals.} Each dispatcher $\disp \in \setD$ has its own stochastic, independent and unknown client request arrival process. We denote by $a^{(\disp)}(t) \in \Natz $ the number of new client requests that exogenously arrive at dispatcher $\disp$ in the first phase of round $t$.
    
  \item \T{Dispatching.} In the second phase of a round, each dispatcher immediately and independently chooses a destination server for each received request and forwards the request to the chosen server queue for processing. We demote by $\receives^{(\disp)}_\serv(t)$ the number of requests dispatcher $\disp$ forwards to server $\serv$ at the second phase of round $t$, and by $\bar{a}_{\serv}(t) =\sum_{\disp\in\setD} \receives^{(\disp)}_\serv(t)$ the total number of requests server $\serv$ receives from all dispatchers.

  \item \T{Departures.} During the third phase, each server performs work and possibly completes requests. Completed requests immediately depart from the system. We denote by $c_\serv(t)$ the number of requests that server $\serv$ can complete during the third phase of round $t$, provided it has that many requests to process. We assume that~$c_\serv(t)$ is determined by an unknown independent stochastic process. We only assume that each server has some inherent expected time invariant processing rate (i.e., speed), and denote $\mathbb{E}[c_\serv(t)]=\rate_\serv$.

\end{enumerate}
Note that any processing system must adhere to the requirement that, on average, the sum of server processing rates must be sufficient to accommodate the sum of arriving client requests. We mathematically express this additional demand in context when appropriate. 
Finally, we use the terms \emph{job}  and \emph{client request} interchangeably.

\section{Solving the Dispatching Problem}
\label{sec:theoretical derivations}
Our goal is to devise an algorithm that leads to short response times and high resource utilization. 
In this section, we define a notion of an ``ideal'' assignment for an online dispatching algorithm.
Intuitively, the quality of a dispatching assignment can be assessed by comparing it with the ideal one. 
A key element in our distributed solution will be solving an optimization problem whose goal is to approximate the ideal assignment as well as possible.

\subsection{Ideally Balanced Assignment}
For each round $t$ we are given the current sizes of the queues at the servers $\big(q_1(t),\ldots,q_n(t)\big)$ and the arrivals at the dispatchers $\big(a^{(1)}(t),\ldots,a^{(m)}(t)\big)$, and we should return an assignment of the incoming jobs to the servers $\big(\receives_1(t),\ldots,\receives_n(t)\big)$. (From here on, we omit the round notation~$t$ when clear from context.)
We would like to distribute the incoming jobs in a manner that balances the {\em load}, which we think of as the amount of work that each server has after the assignment. 
Since servers have different speeds, the load on a server~$\serv$ does not correspond to the number of jobs~$q_\serv + \receives_\serv$  in its queue (at the end of the round). 
Rather, the load is taken to be $\frac{q_\serv+\receives_\serv}{\rate_\serv}$, i.e., the expected amount of time it would take~$\serv$ to process the jobs that are in its queue. (See \Cref{fig:example1} for illustration.)
In general, no assignment that would completely balance the load among all servers will necessarily exist, since server queue-lengths may vary considerably at the start of a round.
We can, however, aim at minimizing the difference between the load of the most and least loaded servers.
If the units of incoming work were continuous, this would be achieved by an assignment $\{\receives_\serv\}_{\serv\in \setS}$ that solves:
\begin{equation}\label{eq:iba problem}
\begin{aligned}
        %\arg 
        \max\min\limits_{\serv\in\setS}%_{\{\receives_1,\ldots,\receives_n\}} 
        \frac{q_\serv + \receives_\serv}{\rate_\serv} \qquad\quad
        \textrm{s.t.} \quad\quad  \sum_{\serv\in\setS} \receives_\serv = \sum_{\disp\in\setD} a^{(\disp)} \text{ \ \ and}\quad
            \receives_\serv \ge 0 \,\, \forall \serv \in \setS. 
\end{aligned}
%\vspace{-3mm}
\end{equation}
%
%%%%%%%%%%%%%%%%%%%%%%%%%%%%%%%%%%%%%%%%%%%%%%%%%%
\begin{figure}[t]
\centering
\begin{subfigure}{.3\textwidth}
  \centering
  \includegraphics[width=0.7\linewidth]{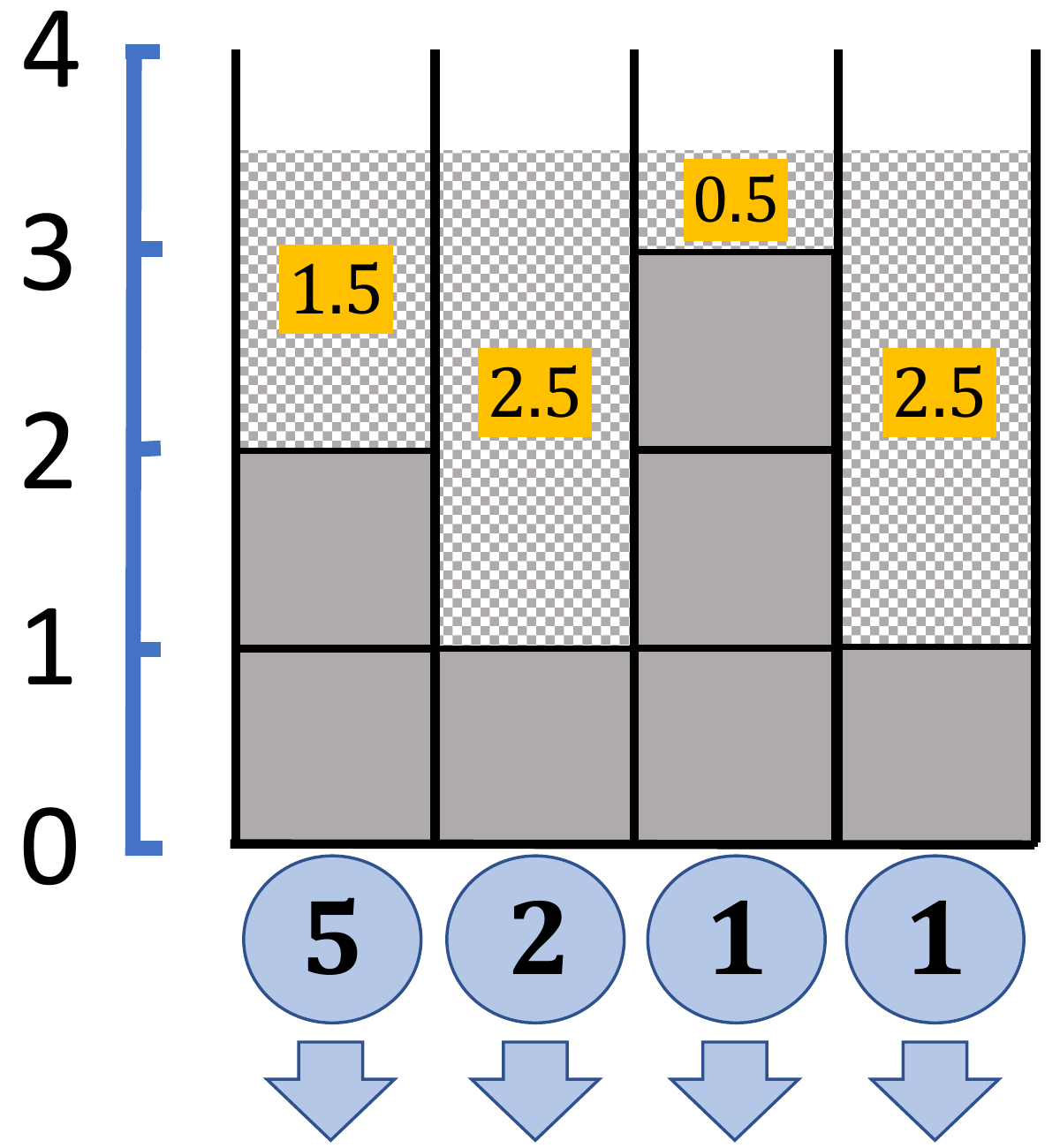}  
  \caption{\footnotesize Ideally balanced queue-lengths.}
  \label{fig:ex1:wl}
\end{subfigure}
\quad\quad
\begin{subfigure}{.58\textwidth}
  \centering
  \includegraphics[width=0.7\linewidth]{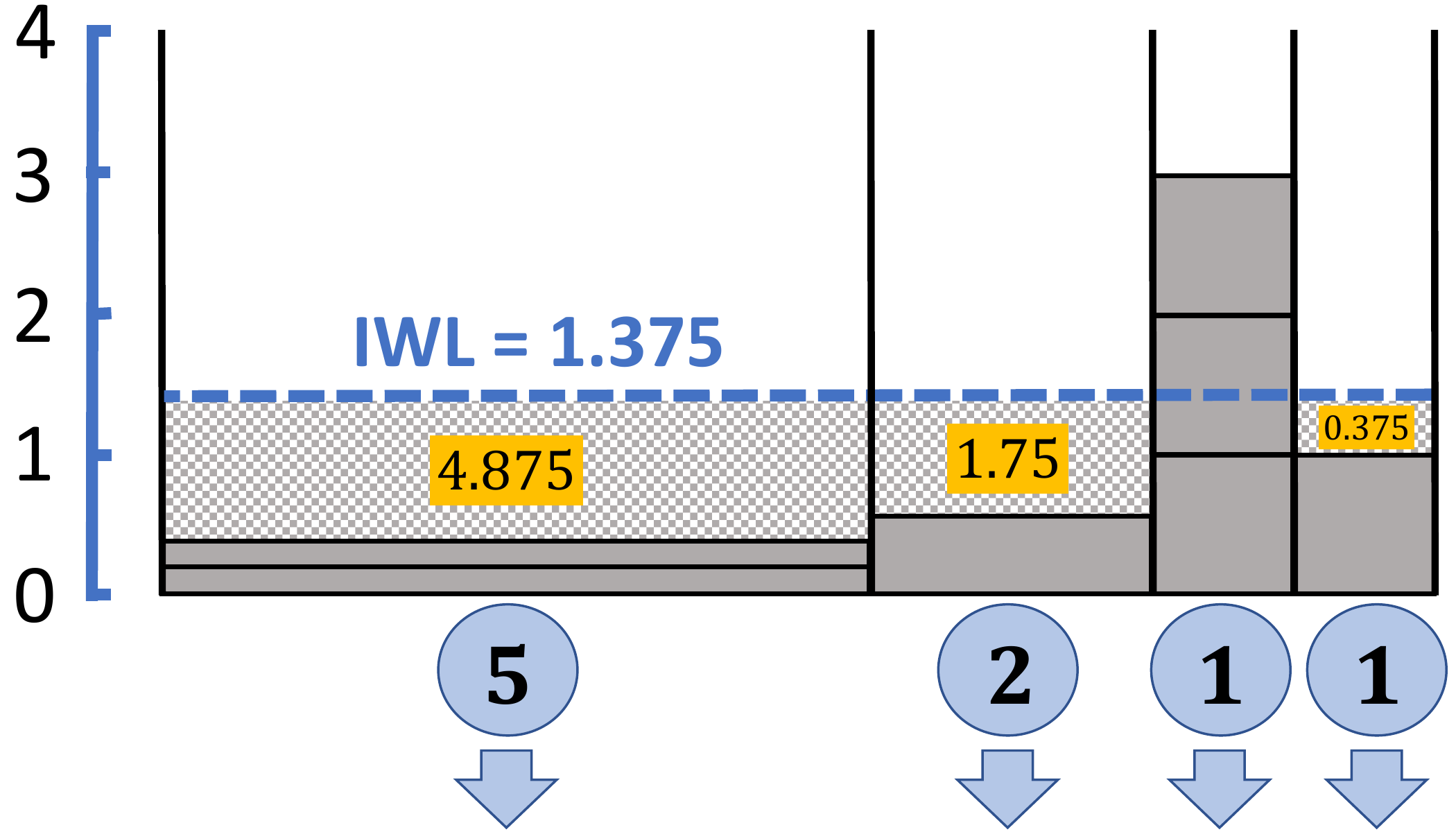}  
  \caption{\footnotesize Ideally balanced workload.}
  \label{fig:ex1:iwl}
\end{subfigure}
\caption{Illustrating the difference between balancing the number of jobs at the servers and balancing the servers' workload. 
An example with 4 servers with rates $[5,2,1,1]$ (from left to right), 7 queuing jobs at the servers $[2,1,3,1]$ and 7 new job arrivals.
An ideally balanced assignment is $[4.875,1.75,0,0.375]$, which differs from $[1.5,2.5,0.5,2.5]$ that balances the number of jobs per server.}
\label{fig:example1}
\end{figure}
%%%%%%%%%%%%%%%%%%%%%%%%%%%%%%%%%%%%%%%%%%%%%%%%%%

We call an assignment that satisfies \cref{eq:iba problem} an  \textit{ideally balanced assignment} (\iba{} for short), and the value of the target function the \textit{ideal workload} (\iwl{}).
An illustrative example for an \iba{} is presented in \Cref{fig:ex1:iwl}.
The \iba{} is an idealized goal in the sense that it is not always possible to achieve.
This is because jobs are discrete and cannot be split among different servers.
Therefore, a realistic load balancing algorithm should strive to assign jobs in a manner that is as close as possible to the \iba{} by some distance measure (e.g., the Euclidean norm).

In a \emph{centralized dispatching system} (with a single dispatcher) minimizing the Euclidean norm distance
from the \iwl{} can be achieved in a straightforward manner.
The dispatcher sends each job to the server that is expected to process it the earliest.
Namely, to the server with the currently minimal $\frac{q_\serv+1}{\rate_\serv}$. 
It then updates the server's queue length and moves on to the next job.
In a \emph{distributed dispatching system} (i.e., with multiple dispatchers), however, the solution is considerably more complex.
In particular, finding the best approximation to the \iba{} in systems with multiple dispatchers requires exact coordination among the dispatchers.
However, dispatchers must make immediate and independent decisions. 
\mbox{They have no time to communicate and so such exact coordination is impossible.}

Following~\cite{DISC2020}, we deal with the need to coordinate dispatchers' decisions while keeping them independent by randomizing their decisions.
In essence, a dispatcher computes a probability distribution $P=[p_1,\ldots,p_n]$ in each round.
Each job's destination is then drawn according to $P$, thus making the decisions independent.
The probabilities in $P$ take into consideration both server loads and what other dispatchers might draw. 
In~\cite{DISC2020}, the authors name this approach stochastic coordination.
The main challenge in stochastic coordination is identifying the right probabilities, and computing them.

We now turn our focus to computing the \iwl.
Note that the \iwl{} determines the \iba{} ($=\{\receives_\serv\}_{\serv\in \setS}$) since for all $\serv\in\setS$:
\begin{equation}\label{eq:iba from iwl}
\begin{aligned}
        \receives_\serv ~=~ \rate_\serv \cdot \max\left\{\frac{q_\serv}{\rate_\serv},\, \iwl \right\} - q_\serv. 
\end{aligned}
\end{equation}
We develop an algorithm that computes the \iwl{} in a heterogeneous system.
The pseudocode appears in \Cref{alg:iba} in \Cref{app:alg:iba}.
Roughly speaking, \Cref{alg:iba} works as follows.
It starts with the current state of the system and iteratively assigns work that increases the minimal load, i.e., $\min\frac{q_\serv+\receives_\serv}{\rate_\serv}$ until it reaches the \iwl.
At each iteration, unassigned work is assigned to the least loaded servers until they reach the closest higher load.
The algorithm returns when no unassigned work remains.
Moreover, \Cref{alg:iba} is efficient. That is, it runs in $O(n)$ time, if the servers are pre-sorted by $\frac{q_\serv}{\rate_\serv}$ (otherwise, its complexity is dominated by the task of sorting these~$n$ values).
In the algorithm and henceforth, we often use `$a$' as a shorthand for the total sum of arrivals:\, $a\,\triangleq \sum\limits_{\disp\in\setD}\!\!{a^{(\disp)}}$.

\subsection{Distributed Load Balancing as a Stochastic Optimization Problem}
Our distributed load balancing algorithm will be based on the solution to a stochastic optimization problem.
The first step in the statement of the optimization problem is to define an appropriate error function that we seek to minimize.
A job assignment is measured against the ideally balanced assignment (\iba).
Assume that at each time slot~$t$ we are given the arrivals $\big(a^{(1)},\ldots,a^{(m)}\big)$ in the current round,%
\footnote{We forego the assumption that the entire vector of arrivals is available in \Cref{sec:putting it all together}.}
and we know the current sizes of the queues at the servers $\big(q_1,\ldots,q_n\big)$.
We can deduce the ideal workload (\iwl) of a centralized \iba{} algorithm at that point (see \Cref{alg:iba}).
Intuitively, a large deviation from the \iwl{} in server $\serv$ corresponds either to large job delays and slow response times (for a higher load than ideal), or lesser resource utilization (for a lower load than ideal).
Since we seek to avoid long delay tails and wasted server capacities, the error for a large deviation from the \iwl{} should be higher than the error for several small deviations.
On the other hand, minimizing the worst-case assignment, that is, focusing too much on the worst possible deviation, can damage the mean performance.
We measure the distance a solution offers from the ideal solution in terms of the $L_2$ norm (squared distances).
This balances the algorithm's mean and worst-case performance and, moreover, is amenable to formal analysis.
Thus, the individual error of server $\serv$ is defined as 
\begin{equation}\label{eq:individual_error}
        error_\serv = %\purple{\left(\frac{\bar{q}_\serv}{\rate_\serv} - \wl\right)^2  =}
        \left(\frac{q_\serv + \receives_\serv}{\rate_\serv} - \wl\right)^2.
\end{equation}
Finally, to account for the variability in processing power, a server's error is weighted by multiplying it by the server's processing speed.
For example, an error of $+1$ in the workload of a server with a processing rate of $\mu_\serv=10$ results in $10$ jobs that will now wait for an expected extra round. 
The same $+1$ error for a server with $\mu_\serv=1$ affects only a single job.
The resulting total error of an assignment $\{ \receives_1, \ldots, \receives_n \}$ is 
\begin{equation}\label{eq:l2_error}
\begin{aligned}
        error = 
            \sum_{\serv\in \setS} \rate_\serv \cdot error_\serv
            =
            \sum_{\serv\in \setS} \rate_\serv \left(\frac{q_\serv + \receives_\serv}{\rate_\serv} - \wl\right)^2
            =\sum_{\serv\in \setS} \frac{(q_\serv + \receives_\serv -\, \rate_\serv\!\cdot\!\wl)^2} {\rate_\serv}.
\end{aligned}
\end{equation}

Recall that we employ randomness to determine job destinations. 
Therefore, an assignment $\{\receives_\serv\}_{\serv\in \setS}$ and the corresponding error are random variables.
In particular, we seek the probabilities $P\triangleq [p_1, \ldots, p_n ]$ that minimize the expected error function given the current queue lengths $[ q_1,\ldots,q_n ]$ and the new arrivals \mbox{$[ a^{(1)},\ldots, a^{(m)} ]$}. Formally,
\begin{equation}\label{eq:optimization problem}
\begin{aligned}
        \arg \min_P \mathbb{E} [error] &= \arg \min_P\mathbb{E} \left[\sum_{\serv\in \setS} \frac{\left(\receives_\serv + (q_\serv - \rate_\serv\wl)\right)^2} {\rate_\serv}\right] \\
        &= \arg \min_P \left( \sum_{\serv\in \setS}\mathbb{E} \left[\frac{\receives_\serv^2}{\rate_\serv}\right]
         + 2\sum_{\serv\in \setS}\mathbb{E} \left[\frac{\receives_\serv\cdot(q_\serv - \rate_\serv\wl)}{\rate_\serv}\right]
         + \sum_{\serv\in \setS}\mathbb{E} \left[\frac{(q_\serv - \rate_\serv\wl)^2}{\rate_\serv}\right] \right).
\end{aligned}
\end{equation}
In the error function only the values of $\{\receives_\serv\}$ are affected by $P$. Therefore, we can drop all additive constants and take the expectation on the  $\{\receives_\serv\}$ terms only. This yields the simplified form
\begin{equation}\label{eq:simplyfied optimization problem 0}
\begin{aligned}
        \arg \min_P \mathbb{E} [error] &= \arg \min_P \left( \sum_{\serv\in \setS} \frac{1}{\rate_\serv} \mathbb{E} [\receives_\serv^2] + 2 \sum_{\serv\in \setS} \frac{q_\serv - \rate_\serv\wl}{\rate_\serv} \mathbb{E} [\receives_\serv] \right).
\end{aligned}
\end{equation}
Since job destinations are drawn independently according to $P$ we have that $\receives_\serv$ is a binomial random variable with
\begin{equation}\label{eq:mean and var}
\begin{aligned}
        \mathbb{E} [\receives_\serv] &= p_\serv\sum_{\disp\in\setD}a^{(\disp)} = a p_\serv  \quad\quad~\textrm{and~~} \quad
        \mathbb{E} [\receives_\serv^2] &= a\cdot p_\serv(1-p_\serv) + a^2 p_\serv^2 .
\end{aligned}
\end{equation}
Plugging the above in \cref{eq:simplyfied optimization problem 0} yields
\begin{equation}\label{eq:simplyfied optimization problem}
\begin{aligned}
        \arg \min_P \mathbb{E} [error] &= \arg \min_P \left( \sum_{\serv\in \setS} \frac{1}{\rate_\serv} \cdot (ap_\serv-ap_\serv^2 + a^2 p_\serv^2) + 2 \sum_{\serv\in \setS} \frac{q_\serv - \rate_\serv\wl}{\rate_\serv} \cdot a p_\serv  \right)\\
        &= \arg \min_P \left( a(a-1)\sum_{\serv\in \setS} \frac{1}{\rate_\serv} \cdot p_\serv^2 + a \sum_{\serv\in \setS} \frac{2(q_\serv - \rate_\serv\wl) + 1}{\rate_\serv} \cdot p_\serv \right).
\end{aligned}
\end{equation}

When $a=1$ the first term is eliminated, and we obtain:
\begin{equation}\label{eq:simplyfied optimization problem a=1}
\begin{aligned}
        \arg \min_P \mathbb{E} [error] =&~ \arg \min_P \sum_{\serv\in \setS} \frac{2(q_\serv - \rate_\serv\wl) + 1}{\rate_\serv} \cdot p_\serv 
        =~ \arg \min_P \sum_{\serv\in \setS} \left( \frac{2q_\serv + 1}{\rate_\serv} -2\wl \right) \cdot p_\serv 
\end{aligned}
\end{equation}
The solution is to divide the probabilities among the servers that have the minimal value of $\frac{2q_\serv + 1}{\rate_\serv}$. (The division can be arbitrary; any division of probabilities will do.)

We turn to the general case in which $a>1$.
Dividing the target function at \cref{eq:simplyfied optimization problem} by $a$, we obtain a constrained optimization problem with the following standard form.
\begin{equation}\label{eq:standard form}
\begin{aligned}
\underset{P}{\textrm{minimize}} \quad & f(P) = (a-1)\sum_{\serv\in \setS} \frac{1}{\rate_\serv} \cdot p_\serv^2 +  \sum_{\serv\in \setS} \frac{2(q_\serv - \rate_\serv\wl) + 1}{\rate_\serv} \cdot p_\serv\\
\textrm{subject to} \quad & \sum_{\serv\in\setS} p_\serv - 1 = 0 ,\\
& p_\serv \ge 0 \,\, \forall \serv \in \setS. 
\end{aligned}
\end{equation}
Using the solution $P=[p_1,\ldots,p_n]$ to the above problem yields an optimal greedy online algorithm, in which the best stochastic per-round decisions are made.
Interestingly, an optimal $P$ can contain positive probabilities even to servers that are above the \iwl.
This is in stark contrast to \cite{DISC2020}, where, underlying the formal analysis is the basic assumption that servers that have a current load of more than the \iwl%
\footnote{The work of \cite{DISC2020} uses a term called \textit{water-level}. In the special case of homogeneous systems, the water level coincides with the \iwl.}
should have 0 probability of receiving any jobs.
\Cref{fig:ex2:scd} provides an illustrative example of what an optimal $P$ is expected to yield.
This example shows an optimal $P$ in which a server that is above the \iwl{} receives a positive probability of $\approx 0.221$ (i.e., $\approx\frac{1.55}{7})$.

The optimization problem in \cref{eq:standard form} is convex quadratic with affine constraints and has a non-empty set of feasible solutions.
Known algorithms computing exact solutions to this problem incur an exponential in $n$ (worst case) time complexity. 
Using well-established algorithms for quadratic programming such as the interior-point method and the ellipsoid method \cite{boyd2004convex}, it is possible to compute approximate solutions in $\Omega(n^3)$. However, this is still not good enough to be used for dispatching in the high-volume settings that we are targeting.

\section{Deriving a Computationally Efficient Solution}
\label{sec:alg comp}
Real-time load balancing systems are reluctant to deploy algorithms that might incur exponential time complexity.
Even the cubic complexity, of the approximation techniques, might become impractical for systems with a few hundred servers.
As we will show, it is possible to identify particular properties of our optimization problem \cref{eq:standard form}, and utilize them to design an algorithm with optimal time complexity.

%%%%%%%%%%%%%%%%%%%%%%%%%%%%%%%%%%%%%%%%%%%%%%%%%%

\begin{figure}[t]
\centering
\begin{subfigure}{.47\textwidth}
  \centering
  \includegraphics[width=\linewidth]{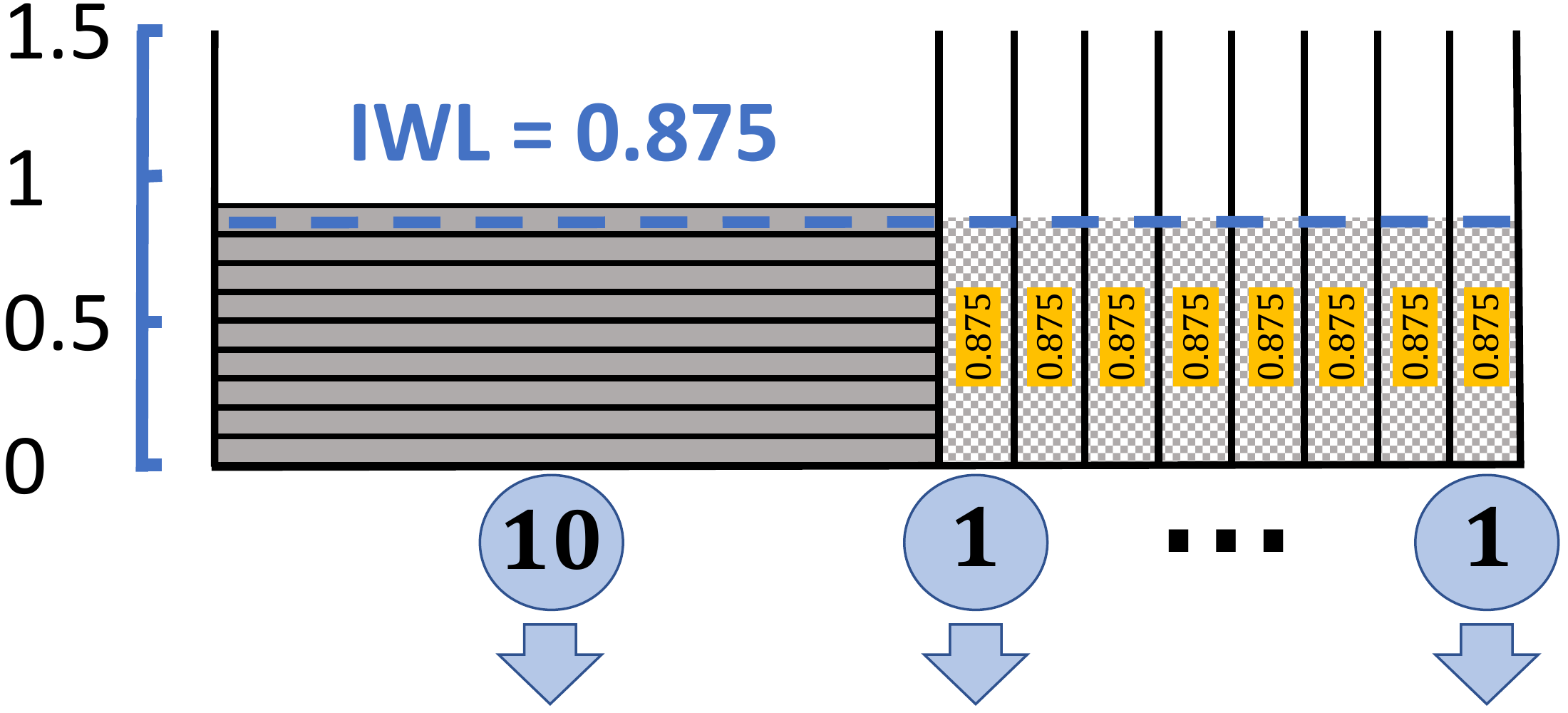}  
  \caption{\footnotesize Ideally balanced workload.}
  \label{fig:ex2:iwl}
\end{subfigure}
\quad\,
\begin{subfigure}{.47\textwidth}
  \centering
  \includegraphics[width=\linewidth]{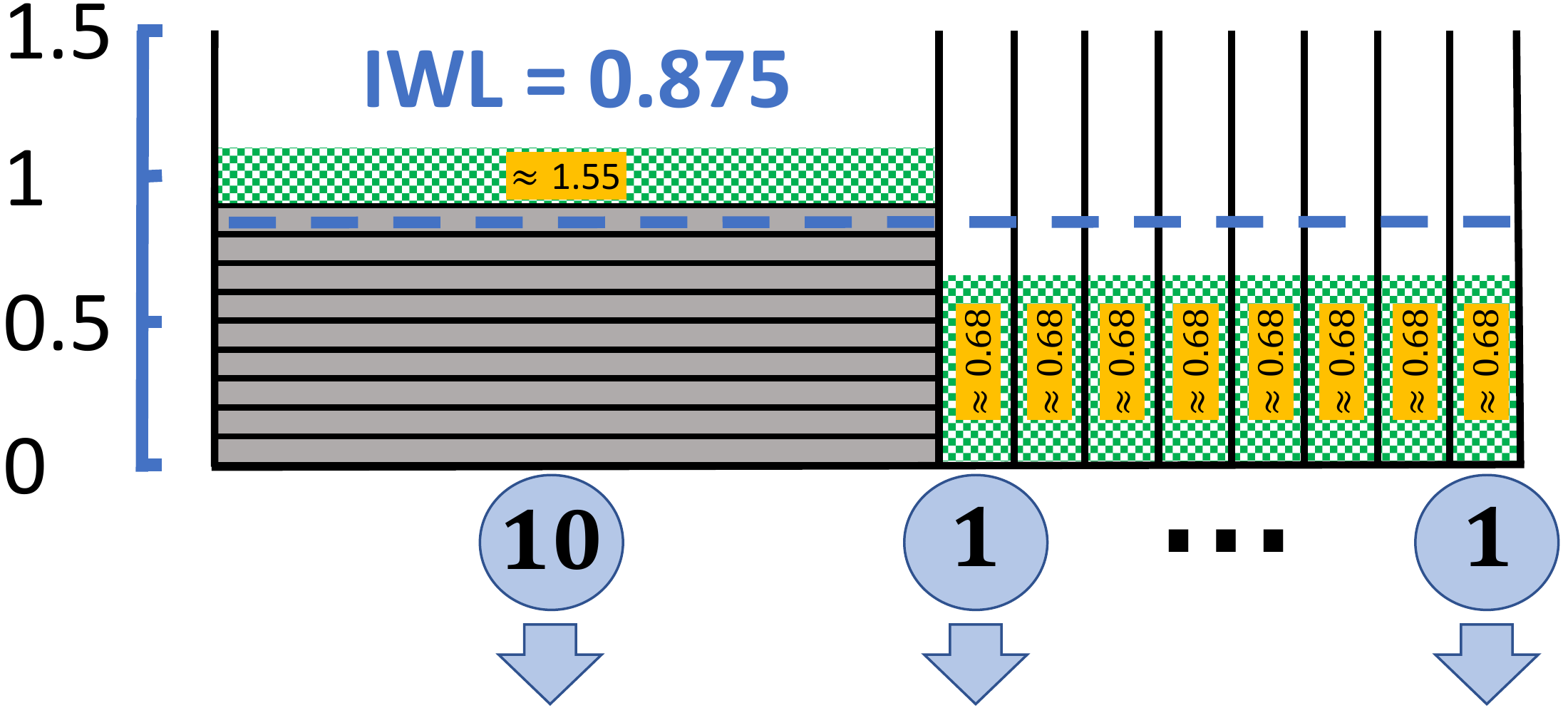}  
  \caption{\footnotesize Expected workload of an optimal $SCD$ dispatching.}
  \label{fig:ex2:scd}
\end{subfigure}
\caption{An example of a system with one fast ($\mu=10$) server, and 8 slower servers ($\mu=1$). The system state is 9 jobs queuing at the fast server, empty queues at the remaining servers, and 7 incoming jobs.}
\label{fig:example2}
\end{figure}

%%%%%%%%%%%%%%%%%%%%%%%%%%%%%%%%%%%%%%%%%%%%%%%%%%

\subsection{The Probable Set and its Ordering}
Recall that $a=1$ corresponds to a single job entering the system in the current round. 
In this case, no coordination between dispatchers is necessary, and indeed, as we have shown, the dispatching problem can be solved in a straightforward manner. 
The distributed problem arises only when $a>1$, in which case we need to solve the optimization problem of \cref{eq:standard form}.   
The  Lagrangian function corresponding to \cref{eq:standard form} is
\begin{equation}
    \begin{split}
        L(P,\Lambda) = &~(a-1)\sum_{\serv\in \setS} \frac{1}{\rate_\serv} \cdot p_\serv^2 +  \sum_{\serv\in \setS} \frac{2(q_\serv - \rate_\serv\wl) + 1}{\rate_\serv} \cdot p_\serv ~-~ \sum_{\serv\in\setS} \Lambda_\serv p_\serv ~+~ \Lambda_0 (\sum_{\serv\in\setS} p_\serv - 1),
    \end{split}
\end{equation}
where $\Lambda_0$ is the Lagrange multiplier that corresponds to the equality constraint and $\set{\Lambda_\serv}_{\serv\in\setS}$ correspond to the inequality constraints.
Since the problem is convex with affine constraints, the Karush-Kuhn-Tucker method (KKT) \cite{karush1939minima,kuhn1951} theorem states that the following conditions are necessary and sufficient for $P^*=[p_1^*,\ldots,p_n^*]$ to be an optimal solution.
For all $\serv\in\setS$:  
\begin{equation}\label{eq:kkt_conditions}
\begin{aligned}
&\frac{\partial L}{\partial p_\serv}(p_\serv^*) = 2(a-1) \frac{1}{\rate_\serv} p_\serv^* + \frac{2(q_\serv - \rate_\serv\wl) + 1}{\rate_\serv} - \Lambda_\serv + \Lambda_0 = 0  
& \text{(Stationarity)} \cr
&\sum_{\serv\in\setS} p_\serv^* - 1 = 0 \,\, \text{and } p_\serv^* \ge 0
&\text{(Primal feasibility)}\cr
& \Lambda_\serv \ge 0 &\text{(Dual feasibility)}\cr
& p_\serv^* \Lambda_\serv = 0
& \text{(Complementary slackness)}
\end{aligned}
\end{equation}

From Stationarity we can deduce that 
\begin{equation}\label{eq:general p_s}
     p_\serv^*  = 
    \frac{- 2(q_\serv - \rate_\serv\wl) - 1 +  \rate_\serv \Lambda_\serv - \rate_\serv \Lambda_0}{2(a-1)}.
\end{equation}
We call the set of servers with positive probabilities in the optimal solution the \textit{probable set} and denote it by $\setSp\!$. 
Formally~$\setSp \triangleq \{\serv\in\setS \mid p_\serv^* > 0 \}$.
As we shall see, $\setSp$ plays an important role in our derivations.
In particular, the Complementary slackness condition from \cref{eq:kkt_conditions} implies that $\Lambda_\serv = 0$ for every $\serv \in \setSp$.
Thus,
\begin{equation}\label{eq:positive p_s}
     p_\serv^*  = 
    \frac{- 2(q_\serv - \rate_\serv\wl) - 1 - \rate_\serv \Lambda_0}{2(a-1)}
    ,\quad \forall \serv\in\setSp.
\end{equation}

We use the Primal feasibility condition from \cref{eq:kkt_conditions} together with \cref{eq:positive p_s} to obtain 
\begin{equation}\label{eq:1 equality}
    1 = \sum_{\serv \in \setS} p_\serv^* 
      = \sum_{\serv \in \setSp} p_\serv^*
      = \sum_{\serv \in \setSp} \frac{- 2(q_\serv - \rate_\serv\wl) - 1 - \rate_\serv \Lambda_0}{2(a-1)},
\end{equation}
from which we can isolate $\Lambda_0$
\begin{equation}\label{eq:Lambda_0}
    \begin{aligned}
        %2(a-1) = &  2\sum_{\serv \in \setSp}(\rate_\serv\wl - q_\serv) - \sum_{\serv \in \setSp}1 - \sum_{\serv \in \setSp}\rate_\serv \Lambda_0\\
         %
        \Lambda_0  ~=~ &  \frac{2\sum\limits_{\serv \in \setSp}(\rate_\serv\wl - q_\serv) - \sum\limits_{\serv \in \setSp}1 - 2(a-1)} {\sum\limits_{\serv \in \setSp}\rate_\serv }.
    \end{aligned}
\end{equation}

The KKT conditions enabled us to derive \cref{eq:positive p_s} and \cref{eq:Lambda_0}, which show that identifying $\setSp$ provides an analytical solution to our optimization problem \cref{eq:standard form}.
(By first calculating $\Lambda_0$ and then each of the positive probabilities.)
Hence, finding an optimal solution reduces to finding the probable set.
This result is a complementary instance of the generic active set method for quadratic programming \cite{nocedal2006numerical}.
However, the active set method provides, in general, a worst case running-time complexity of $2^n$.
And while for the specific instance of the problem with $\rate_\serv=\rate$ for all servers (i.e., an homogeneous system) we can derive $\setSp = \{\serv\in\setS \mid \frac{q_\serv}{\rate_\serv} < \iwl \}$ analytically, the example in \Cref{fig:ex2:scd} shows that this is no longer true in the general, heterogeneous, case.

Instead of deriving $\setSp$ analytically, we turn to find an algorithmic solution. 
A trivial algorithm is to examine each of the $2^{n}$ possible subsets of servers. 
For each candidate subset: first calculate $\Lambda_0$ according to \cref{eq:Lambda_0} --- this guarantees that the sum of probabilities is 1, then test whether all the probabilities are indeed positive, and finally calculate the objective function.
Clearly, its exponential computation complexity renders this method infeasible.
To overcome the exponential nature of searching in a domain of size $2^n$, we must reduce the size of the domain we search in.
\Cref{lem:ordering} formulates a property of the objective function in \cref{eq:standard form} that holds the key to reducing the size of the search.

\begin{restatable}{lemma}{orderingLemma}
\label{lem:ordering}
    Let~$\servT$ be a server in the probable set~$\setSp$. For every server~$\servU\in\setS$, if  $~~\frac{2q_{\servT} + 1}{\rate_{\servT}} \ge \frac{2q_{\servU} + 1}{\rate_{\servU}}$ then $\servU\in\setS^{+}$ as well.
\end{restatable}

\begin{proof}[Proof Sketch]
    The full proof is deferred to \Cref{app:proofs}. 
    Here we provide a proof sketch.
    Let $\servT$ and $\servU$ satisfy the assumption, and let $P^*=\{p^*_1,\ldots, p^*_n\}$ be the optimal solution for \cref{eq:standard form}.
    Assume by contradiction that $p^*_{\servT}>0$ but $p^*_{\servU}=0$.
    We show that there exists a feasible solution $P$ that obtains a lower value of the objective function, thus contradicting the optimality of $P^*$.
    Specifically, we show that for the positive constant~$\const = \min \{\frac{\rate_{\servU} (2q_{\servT} + 1) - \rate_{\servT} (2q_{\servU} + 1) + 2\rate_{\servT} \rate_{\servU} p^*_{\servT}} {(a-1)(\rate_{\servT} + \rate_{\servU}) }, p^*_{\servT}\}$, any $0< \epsilon < z$ and a solution $P=\{p_1, \ldots, p_n\}$ with $p_{\servU}=\epsilon, p_{\servT}=p^*_{\servT}-\epsilon$, and $p_\serv=p^*_\serv$ for all other servers, it holds that $P$ is a better feasible solution.
\end{proof}

\Cref{lem:ordering} imposes a very strict constraint on the structure of the probable set~$\setSp$: 
\begin{corollary}\label{corollary:n_subsets}
    Let $\serv_{i_1}, \serv_{i_2},\ldots,\serv_{i_n}$ be a listing of the  servers in~$\setS$ in non-decreasing sorted order of $\frac{2q_{\serv} + 1}{\rate_{\serv}}$. Moreover, denote $\setS_j=\{\serv_{i_1},\ldots,\serv_{i_j}\}$ for $j=1,\ldots,n$.  Then   $\setS^{+}=\setS_j$  ~~for some $j\le n$.
\end{corollary}
\Cref{corollary:n_subsets} allows us to find the optimal probabilities in polynomial time.
To be precise, both sorting $\setS$ and computing the \iwl{} take $O(n\log n)$ time. So we are left with computing for each of the $n$ subsets~$\setS_j$: (1)~whether it satisfies that all probabilities are non-negative, \ and (2)~What the value of the subset's objective function is.
Steps (1) and~(2) can be implemented in $O(n)$ time complexity each.
Finally, we extract the subset with the minimal objective function value from those that respect (1).
This can be done in $O(n^2)$ complexity, as presented in \Cref{alg:n2}.
%%=============================%%
\begin{algorithm}[ht]
  \DontPrintSemicolon
  \SetKwFunction{FMain}{ComputeProbabilities}
  \SetKwProg{Pn}{Function}{:}{\KwRet}
  \Pn{\FMain{$\setS$, $\setQ$, $\setDu$, $a$, \wl}}{
        $\setO \gets \emptyset$;\,\,$\val^* \gets \infty$\;
        $p_\serv\gets0\,\forall\serv\in\setS$\;
        \While{$\setS{\setminus}\setO \neq \emptyset$}
        {
            $\curr \gets \arg\min_{\serv\in \setS\setminus\setO} \frac{2q_\serv+1}{\rate_\serv}$\;
            $\setO \gets \setO \cup \curr$ \tcp*[r]{$\setO$ is the candidate set for $\setSp$}
            %$P[\serv] \gets 0 \quad \forall \serv \in \setS$ \tcp*[r]{all probabilities are initialized to 0.}
            $\Lambda_0 \gets  \frac{2\sum_{\serv \in \setO}(\rate_\serv\wl - q_\serv) - \sum_{\serv \in \setO}1 - 2(a-1)} {\sum_{\serv \in \setO}\rate_\serv }$
            \tcp*[r]{according to \cref{eq:Lambda_0}}
            \For{$\serv\in\setO$}
            {
                $p_\serv \gets \frac{- 2(q_\serv - \rate_\serv\wl) - 1 - \rate_\serv \Lambda_0}{2(a-1)}$ \tcp*[r]{according to \cref{eq:positive p_s}}
                \If{$p_\serv<0$}
                    {go to line~4 \tcp*[r]{the solution is infeasible; continue to next $\curr$}}
                    %$P[\serv] \gets p_\serv$\;
            }
            $\val \gets (a-1)\sum_{\serv\in \setO} \frac{1}{\rate_\serv} \cdot p_\serv^2 +  \sum_{\serv\in \setO} \frac{2(q_\serv - \rate_\serv\wl) + 1}{\rate_\serv} \cdot p_\serv$ \tcp*[r]{according to \cref{eq:standard form}}
            \If{$\val < \val^*$} 
            {    
                $\val^* \gets \val$\;
                $P^*\gets \{p_\serv\}_{\serv\in\setS}$\;
            }              
        }
        \Return $P^*$\;
    }
    \caption{Find probabilities in $O(n^2)$ time.}
    \label{alg:n2}
    
\end{algorithm}
%%=============================%%

\subsection{Optimal Complexity}

By reduction from sorting, we have that any solution that induces a total order on the servers (including $JSQ$, $SED$ and $TWF$) must incur a complexity of $\Omega(n\log n)$.
The quadratic complexity of \Cref{alg:n2} is caused by the $\Omega(n)$ cost per iteration of the calculation in line~7, the {\bf for} loop in lines~8-11, and the summation in line~12. 
Since the outcome of an iteration depends on the results of past iterations, we can employ a dynamic programming approach and design an algorithm with an $O(n\log n)$, and hence optimal, complexity. 
We do this in \Cref{alg:opt}, whose pseudocode is presented in \Cref{app:alg:opt}.
We next explain the necessary steps in deriving \Cref{alg:opt}.

Replacing the calculation of $\Lambda_0$ in line~7 of~\Cref{alg:n2} by an $O(1)$ computation per iteration is quite straightforward.
We simply calculate the enumerator and denominator sums separately, and then divide them.
Each sum is computed by adding the current element to the sum from the previous iteration.
The {\bf for} loop in lines~8-11 of \Cref{alg:n2} tests whether the computed probabilities are indeed non-negative.
Since the denominator is a positive constant, it is sufficient to test the enumerator for each server in $\setO$. That is, whether it holds that $- 2(q_\serv - \rate_\serv\wl) - 1 - \rate_\serv \Lambda_0 \ge 0$.
% \begin{equation}\label{eq:pos nom}
%     \begin{aligned}
%         - 2(q_\serv - \rate_\serv\wl) - 1 - \rate_\serv \Lambda_0 &\ge 0.
%     \end{aligned}
% \end{equation}
The rates $\rate_\serv$ of the servers are positive. Therefore, dividing by $\rate_\serv$ and rearranging yields
\begin{equation}\label{eq:pos nom simp}
    \begin{aligned}
        \iwl - \Lambda_0 &\ge \frac{2 q_\serv + 1}{\rate_\serv}.
    \end{aligned}
\end{equation}
In turn, observe that \cref{eq:pos nom simp} holds for all servers in $\setO$ iff it holds for the server with the highest $\frac{2 q_\serv + 1}{\rate_\serv}$ value in~$\setO$, which is server~$\servT$ in each iteration.
We thus replace the $\Omega(n)$ complexity \textbf{for} loop of lines~8-11 in \Cref{alg:n2} by this single $O(1)$ complexity test.
Next, we address the summation in line~12 of \Cref{alg:n2} which computes the value of the objective function $f(P)$ from \cref{eq:standard form}.
Since $P$ contains $\Theta(n)$ elements, computing $P$ costs $\Omega(n)$ time.
However, we can make the computation of $P$ more efficient by using the following Lemma. 
\begin{restatable}{lemma}{lemTwo}
\label{lem:2}
The objective function in \cref{eq:standard form} satisfies
\begin{equation*}
\begin{aligned}
    f(P)
    =~~~ \Lambda_0^2\underbrace{\sum_{\serv \in \setSp} \frac{\rate_\serv}{4(a-1)}}_{v_1}
     ~~-~~ \underbrace{\sum_{\serv \in \setSp} \frac{(2(q_\serv - \rate_\serv\wl) + 1)^2 }{4\rate_\serv(a-1)}}_{v_2}
    .
\end{aligned}
\end{equation*}
\end{restatable}
\begin{proof}
  See \Cref{app:proofs}.  
\end{proof}
\Cref{lem:2} enables to compute $v_1$ and $v_2$ in $O(1)$ operations per iteration, and reduce the $\Omega(n)$ cost of line~12 in \Cref{alg:n2} to $O(1)$.
\mbox{It follows that \Cref{alg:opt} runs in $O(n \log n)$, which is optimal.}

\section{Putting It All Together}
\label{sec:putting it all together}
We now show how the algorithms from \Cref{sec:theoretical derivations,sec:alg comp}, which compute the \iwl{} and the optimal probabilities, can be employed in the complete dispatching procedure given in \Cref{alg:scd}.
The complexity of an individual dispatcher's computation in any given round is $O(n\log n)$, and it is dominated by the sorting of~$n$ values in lines~2 and~3.
If the sorted order is available to both \Cref{alg:iba} and \Cref{alg:opt}, their running time is reduced to $O(n)$.
This is useful since (1) a designer may implement a sorted data structure in various ways according to what benefits her specific system's characteristics, and~(2) if queue-length information is available before the job arrivals, the server ordering can be precomputed, further improving the online complexity.

\subsection{Estimating the arrivals}

While the derivation of the optimal probabilities assumed knowledge of the arrivals $\{ a^{(1)},\ldots, a^{(m)} \}$, a dispatcher~$\disp$ only knows its own arrivals, i.e., $a^{(\disp)}$.
Nevertheless, we note that the optimal probabilities 
depend only on the total number $a \triangleq \sum\limits_{\disp\in\setD}a^{(\disp)}$ of jobs that arrive, and not on the individual values $a^{(\disp)}$.
We are thus left with estimating $a$.
There are numerous optional methods for estimating~$a$ (e.g., assuming that it is the maximal capacity of the system, ML estimators, etc.).
Following simple and elegant approach taken in \cite{DISC2020}, we have dispatcher $\disp$ estimate $a$ by assuming that everyone else is receiving the same number of jobs as~$\disp$:
\begin{equation}\label{eq:estimating a}
\begin{aligned}
        a_{\text{est},\disp} = m\cdot a^{(\disp)}.
\end{aligned}
\end{equation}
With this, \Cref{alg:scd} is fully defined and can be implemented.

One reason why such an estimation  scheme is effective is because the average estimation of the dispatchers exactly equals the total arrivals.
That is,
\begin{equation}\label{eq:average estimation}
\begin{aligned}
        \frac{1}{m}\cdot \sum_{\disp\in\setD} a_{\text{est},\disp} = \frac{1}{m}\cdot \sum_{\disp\in\setD} m a^{(\disp)} = \sum_{\disp\in\setD} a^{(\disp)} = a.
\end{aligned}
\end{equation}
Consequently, if some dispatchers overestimate the \iwl{} and therefore assign smaller probabilities to less loaded servers, then other dispatchers underestimate the \iwl{} and increase these probabilities. Roughly speaking, these deviations compensate for one another.
In \cref{sec:evaluation} we show how this simplistic estimation technique results in consistent state-of-the-art performance across many systems and metrics.

\subsection{Stability}

A formal guarantee for  dispatching algorithms that is often considered desirable in the literature is called \emph{stability}~(see \Cref{app:stabiltiy}).
Under mild assumptions on the stochastic nature of the arrival and service processes, stability ensures that the servers' queue-lengths will not grow unboundedly so long as the arrivals do not surpass all servers' total processing capacity. 
Accordingly, in Appendix~\ref{app:stabiltiy} we make the appropriate formal definitions and prove that $SCD$ is stable. Moreover, our stability proof holds not only when applying the estimation technique in \cref{eq:estimating a}, but also applies for any estimation technique in which $1 \le a_{\text{est},\disp} < \infty$. Intuitively, this is because for $a_{\text{est},\disp}=1$, our policy behaves similarly to $SED$, and as $a_{\text{est},\disp}\to\infty$ it approaches weighted-random. 
However, with any reasonable estimation (and in particular the one in \cref{eq:estimating a} that we employ), the $SCD$ procedure finds the best of both worlds.
It eliminates the herding phenomenon incurred by $SED$, while also sending sufficient work to the less loaded servers, unlike the load oblivious weighted-random.

\begin{algorithm}[ht]
  \DontPrintSemicolon
  \SetKwFunction{FMain}{Global}
  \SetKwProg{Pn}{Function}{:}{\KwRet}
  \SetKwInOut{Input}{input}
  \SetKwFor{UponKW}{upon}{do}{fintq}

    \Input{$\setS$, $\setDu$}
    \For{each round $t=1,2,\ldots$}{
    Update queue-lengths $\setQ$\;
    $\setS_1\gets$ Sort $\setS$ according to $\frac{q_\serv}{\rate_\serv}$\;
    $\setS_2\gets$ Sort $\setS$ according to $\frac{2q_\serv+1}{\rate_\serv}$\;
    
    \UponKW{receiving jobs $a^{(\disp)}$}{
        Estimate $a$ by $a_{\text{est},\disp}$\;
        $\wl\gets \textsc{ComputeIdealWorkLoad}(\setS_1, \setQ, \setDu, a_{\text{est},\disp})$\;
        $P\gets \textsc{ComputeProbabilities}(\setS_2, \setQ, \setDu, a_{\text{est},\disp}, \wl)$\;
        \For{each job in $a^{(\disp)}$}{
            draw a server $\serv$ according to $P$\;
            send job to server $\serv$\;
        }
    }
    }
    \caption{\SCD{} (SCD) Procedure.}
    \label{alg:scd}
\end{algorithm}
%%=============================%%

\section{Evaluation}
\label{sec:evaluation}

In this section, we conduct simulations and compare our solution to both well-established and to most recent state-of-the-art algorithms in the heterogeneous multi-dispatcher setting. The performance of load balancing techniques is usually evaluated at high loads, so that the job arrival rates are sufficiently high to utilize the servers. 

In such a setting, we are interested in two main performance criteria. The first criterion is the average response time over all client requests (i.e., jobs). For each request, we measure the number of rounds it spent in the system (i.e., from its arrival to a dispatcher to its departure from the server that processed it). The second criterion is the response times' tail distribution. This metric is crucial since it often represents client experience due to two reasons:  (1) clients may send multiple requests (e.g., browsing a website), and delaying a small subset of these may ruin client experience; (2) a client request may be broken into multiple smaller tasks and the request time will be determined by the last tasks to complete (e.g., search engines). In both cases, it is important to meet the desired tail latency for the 95th, 99th, or even the 99.9th percentile of the distribution \cite{dean2013tail,nishtala2017hipster,DISC2020}.    
%\yoram{should we also cite \cite{DISC2020} here?}  

\subsection{Setup} 

In our simulations, every run of each algorithm lasts for $10^5$ rounds. A round is composed of three phases as described in Section \ref{sec:Model}. Namely, first new requests arrive at the dispatchers. The dispatchers must then act immediately and independently from each other and dispatch each request to a server for processing. Finally, the servers perform work and possibly finish requests. Finished requests immediately depart from the system.  

Each dispatcher has its own arrival process of requests. Specifically, the number of requests that arrive to each dispatcher $\disp$ at each round $t$ is drawn from a Poisson distribution with parameter $\lambda$. Formally, ${a^{(\disp)}(t)\sim Pois(\lambda_\disp)}$. 
Each server has its own processing rate. Specifically, the number of requests that server $\serv$ can process in each round $t$ is geometrically distributed with parameter $\rate_\serv$. Formally, $c_\serv(t) \sim Geom(\frac{1}{1+\rate_\serv})$. 

Clearly, when the average number of requests that arrive at the system surpasses the average processing rate of all servers combined, the system is considered \emph{infeasible}, i.e., it cannot process all requests and must drop some of them. It is therefore a standard assumption that the system is \emph{admissible}~\cite{mitzenmacher2000useful,zhou2020asymptotically,vargaftik2020lsq,wang2018distributed,stolyar2017pull,mitzenmacher2016analyzing,stolyar2015pull,lu2011join,haproxy_po2,ngynx_po2,DISC2020}, i.e., that the total average arrival rate at the system is upper-bounded by the total processing capacity of all servers. Accordingly, this is a setting we are interested in examining.  
We define the {\em offered load} to be $$\rho=\frac{\mathbb{E}[\sum_{\disp\in\setD} a^{(\disp)}(0)]}{\mathbb{E}[\sum_{\serv\in\setS} c_\serv(0)]} = \frac{\sum_{\disp\in\setD} \lambda_\disp}{\sum_{\serv\in\setS} \mu_\serv}.$$ 
To be admissible, it must hold that $\rho<1$. Therefore, in all our simulations we test the performance of the different dispatching algorithms for $\rho\in(0,1)$.

We have implemented 10 different dispatching techniques in addition to ours. These include both well-established techniques and the most recent state-of-the-art techniques. In particular, we compare against $JSQ$~\cite{weber1978optimal,winston1977optimality,eryilmaz2012asymptotically}, $SED$~\cite{gardner2021scalable,jaleelgeneral,gardner2019smart,selen2016steady}, $JSQ(2)$~\cite{luczak2006maximum,vvedenskaya1996queueing,mitzenmacher2001power}, $WR$ (weighted random), $JIQ$~\cite{lu2011join,mitzenmacher2016analyzing,stolyar2017pull,stolyar2015pull,van2017load}, $LSQ$~\cite{vargaftik2020lsq,zhou2020asymptotically}, as well as $hJSQ(2)$, $hJIQ$ and $hLSQ$. The last three policies, i.e., $hJSQ(2)$, $hJIQ$ and $hLSQ$ are the adaptations of $JSQ(2)$, $JIQ$ and $LSQ$ to account for server heterogeneity.%
\footnote{Similarly to $SED$, the servers are ranked by their expected delay, i.e., by $\frac{q_\serv}{\mu_\serv}$, instead of by $q_\serv$. Likewise, when random sampling of servers occurs, servers are sampled proportionally to their processing rates rather than uniformly. Specifically, the probability to sample server $\serv$ is $\frac{\mu_\serv}{\sum_{\serv\in\setS} \mu_\serv}$ instead of $\frac{1}{n}$.}
We also compare against the recent $TWF$ policy of~\cite{DISC2020} that achieves stochastic coordination for homogeneous systems. For a fair comparison, in all our experiments, we use the same random seed across all algorithms, resulting in identical arrival and departure processes.  

For clarity, in the paper's main body, we present only the 6 most competitive algorithms out of the 10 implemented. We refer the reader to Appendix \ref{app:Response time} for additional simulation results that compare between $SCD$ and $JSQ(2)$, $JIQ$, $LSQ$ and $WR$ (weighted random). In a nutshell, the $JSQ(2)$, $JIQ$ and $LSQ$ algorithms are less competitive since they do not account for server heterogeneity, whereas $WR$ completely ignores queue length information and fails to leverage less loaded servers in a timely manner.

%%% Moderate heterogeneity
\begin{figure}[t]
    \centering
    \begin{subfigure}{0.84\textwidth}
      \centering
      \caption{Average response time.}
      %\vspace{-2mm}
    \includegraphics[width=\linewidth]{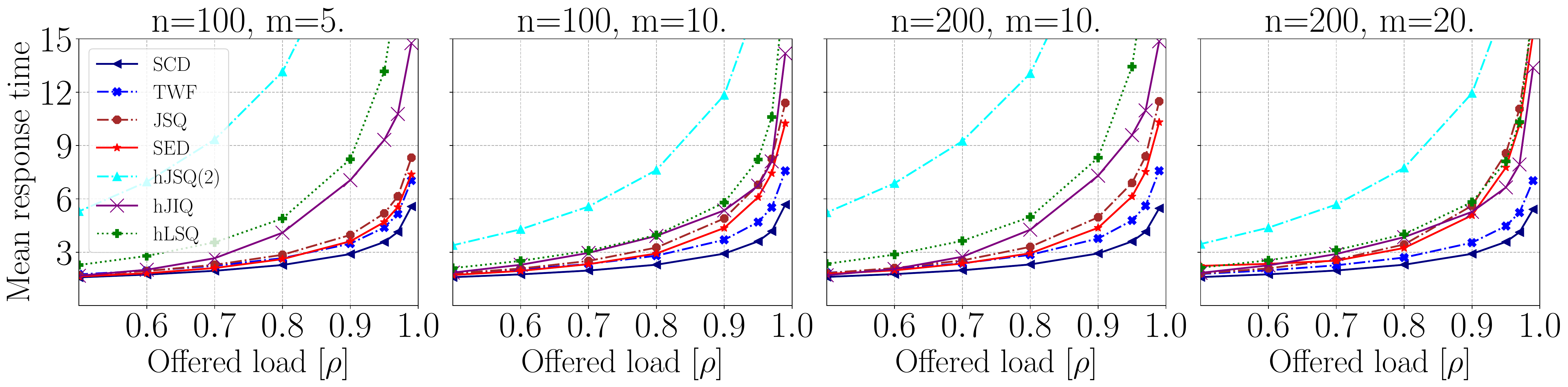}
     \label{fig:evaluation:moderate:loadsweep}
    \end{subfigure}
    %\vspace{-6mm}
    \begin{subfigure}{0.84\textwidth}
      \centering
      \caption{Response time delay tail.}
      %\vspace{-2mm}
    \includegraphics[width=\linewidth]{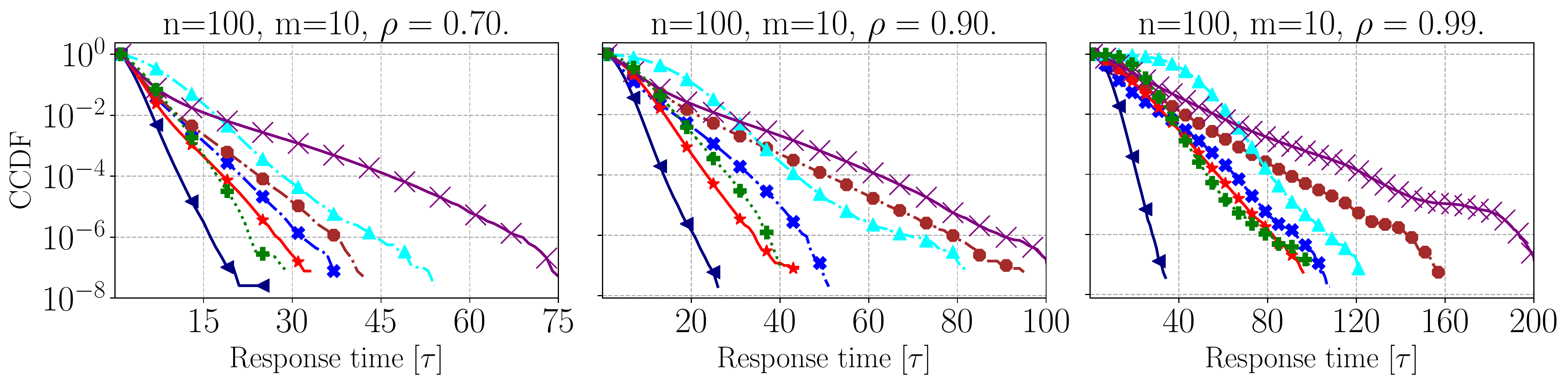}  
     \label{fig:evaluation:moderate:delaytail}
    \end{subfigure}
    \caption{ The service rate (i.e., speed) of each server is randomly drawn from the real interval [1,10]. (a) Average request response time as a function of the offered load over four different systems. The $x$-axis represents the offered load $\rho$. The $y$-axis represents the average response time in number of rounds. (b) Response-time tail distribution over a system with 100 servers and 10 dispatchers over three different offered loads. The $x$-axis represents the response time in number of rounds (denoted by $\tau$). The $y$-axis represents the complementary cumulative distribution function (CCDF). }
    \label{fig:evaluation:moderate}
\end{figure}

%%% High heterogeneity
\begin{figure}[t]
    \centering
    \begin{subfigure}{0.84\textwidth}
      \centering
      \caption{Average response time.}
      %\vspace{-2mm}
    \includegraphics[width=\linewidth]{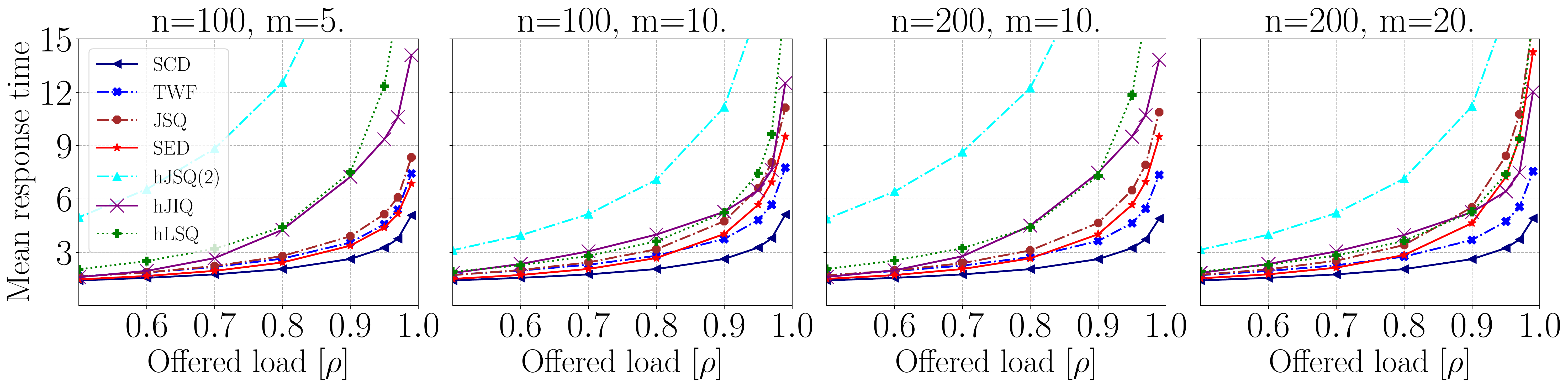}
     \label{fig:evaluation:high:loadsweep}
    \end{subfigure}
    \begin{subfigure}{0.84\textwidth}
      \centering
      \caption{Response time delay tails.}
      %\vspace{-2mm}
    \includegraphics[width=\linewidth]{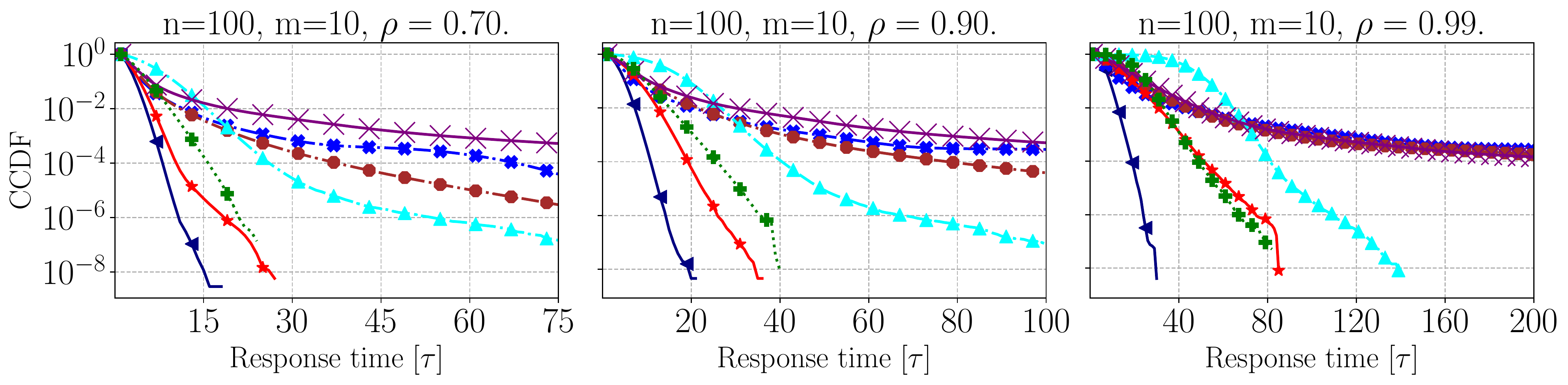}  
     \label{fig:evaluation:high:delaytail}
    \end{subfigure}
    \caption{ The service rate of each server is randomly drawn from the real interval [1,100]. }
    \label{fig:evaluation:high}
\end{figure}

\subsection{Response time}\label{sec:eval:Response time}

We ran the evaluation over four different systems and two different server heterogeneity configurations to represent: (1) moderate heterogeneity that may appear by having different generations of CPUs and hardware configurations of servers and, (2) higher heterogeneity settings that may appear in the presence of accelerators (e.g., FPGA or ASIC).

Figure \ref{fig:evaluation:moderate} shows the evaluation results for case (1) where $\mu_\serv \sim U[1,10]$. Namely, the service rate of each server in each system is randomly drawn from the real interval [1,10]. Figure \ref{fig:evaluation:moderate:loadsweep} shows the average response time of requests as a function of the offered load at the system. It is evident that $SCD$ consistently achieves the best results across all systems and offered loads. Figure \ref{fig:evaluation:moderate:delaytail} shows the response time tail distribution in a system with 100 servers and 10 dispatchers. Again,  $SCD$ achieves the best results with no clear second best. 
For example, at the offered load of $\rho=0.99$ and considering the $10^{-4}$ percentile, which is often of interest, $SCD$ improves over the second-best algorithm ($hLSQ$ in this specific case) by over 2.1$\times$. $TWF$, which is the second-best at the average response time, degrades here and is outperformed by $SED$ and $hLSQ$ since they account for server heterogeneity where $TWF$ does not.

Figure \ref{fig:evaluation:high} shows the evaluation results for case (2) where $\mu_\serv \sim U[1,100]$.
For the average response time metric, Figure \ref{fig:evaluation:high:loadsweep} shows similar trends to case (1), where $SCD$ consistently offers the best results.
Likewise, for the response time tail distribution, Figure \ref{fig:evaluation:high:delaytail} depicts that $SCD$ achieves the best results by an even larger margin than in case (1). For example, at the load of $\rho=0.99$ and considering $10^{-4}$ percentile, $SCD$ improves over the second-best algorithm by over 2.3$\times$.
Note that with this higher heterogeneity, the delay tail distribution of $TWF$ and $JSQ$ is significantly degraded (by more than an order of magnitude even for a load of $\rho=0.7$).

\subsection{Execution run-time}\label{sec:eval:Execution running times}

We next test $SCD$'s execution running times. That is, given the system state and arrivals, how much time does it take for a dispatcher to calculate the dispatching probabilities for that round? 
To answer that question, we implemented all dispatching techniques, and in particular $SCD$ using algorithms \ref{alg:n2} and \ref{alg:opt} as well as $JSQ$ and $SED$ in C++ and optimized them for run-time purposes. 
All running times were measured using a single core setup on a machine with an Intel Core i7-7700 CPU @3.60GHz and 16GB DDR3 2133MHz RAM.

For each algorithm in each round, we measure the time it takes each dispatcher to calculate its requests assignment to the servers. While the asymptotic complexity of the algorithms is fixed, the complexity of each particular instance may require a different computation, depending on the number of arrived requests and server queue-lengths. Therefore, we report the cumulative distribution function (CDF) of those times.

Figure \ref{fig:evaluation:speed} shows the results for the setting described (i.e., $\mu_\serv \sim U[1,10]$).
The running time of $SCD$ via Algorithm \ref{alg:opt} scales similarly to $JSQ$ and $SED$ as expected. I.e., all three have the complexity of $O(n\log n)$. $SCD$ via Algorithm \ref{alg:n2}, on the other hand, is slower. These running time  measurements prove to be consistent. In particular, we obtained similar results for other systems and server heterogeneity levels.
Additional results appear in Appendix \ref{app:Execution run-time}.
We conclude that from the algorithm running time  perspective, $SCD$, when implemented via algorithm \ref{alg:opt}, incurs an acceptable computational overhead, similar to that of $JSQ$ and $SED$, which are in use in \mbox{today's high-performance load balancers \cite{haproxy_gen,ngynx_wlc}.}  

%%% Execution run-time
\begin{figure}[t]
    \centering
    \includegraphics[width=0.92\linewidth]{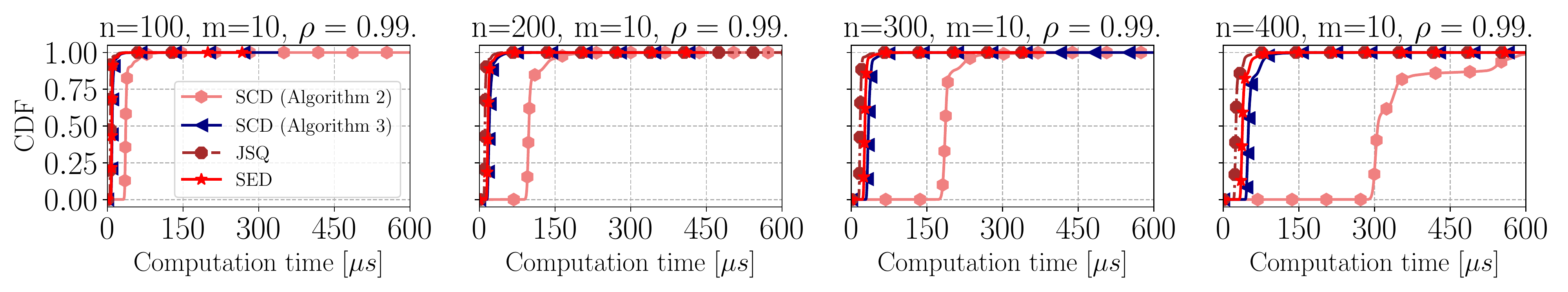} 
    \caption{ Evaluation of execution running times over systems with an increasing number of servers and $\mu_\serv \sim U[1,10]$. }
    \label{fig:evaluation:speed}
\end{figure}

\section{Discussion}\label{sec:disc}
Large scale computing systems are ubiquitous today more than ever.
Often, these systems have many distributed aspects which substantially affect their performance.
The need to address the distributed nature of such systems, both algorithmically and fundamentally, is imperative.
In this work, we presented $SCD$, a load balancing algorithm that addresses modern computer clusters' distributed and heterogeneous nature in a principled manner.
Extensive simulation results demonstrate  that $SCD$ outperforms the state-of-the-art load balancing algorithms across different systems and metrics.
Regarding computation complexity, we designed $SCD$ to run in optimal $O(n\log n)$ time. Therefore, it is no harder to employ in practical systems than traditional approaches such as $JSQ$ and $SED$.

Our work leaves several open problems. For example:
(1) The amount of work that a job requires may depend on specific features of the server processing it. 
%Jobs may entail individual distribution on the amount of work they require, depending on the recipient server's identity. 
%I.e., heterogeneity of jobs.
Can information about the nature of jobs and features of servers be used to further 
%How to use this information to 
improve the stochastic coordination among the dispatchers? 
(2) It may be that not all dispatchers are simultaneously connected to all servers. 
This implies that the distributed nature of the system might not be symmetric.
How should we incorporate such (possibly dynamic) connectivity information while maintaining stochastic coordination?
%(3) Modern applications often rely on auto-scaling to automatically add and remove servers to efficiently accommodate incoming jobs. When a new server is added, we would like to balance its load to correspond to the loads of other servers. Traditional approaches fail to do so.\footnote{For example, $JSQ$ and $SED$ will forward all their incoming requests to the new server resulting in herding, while in other policies such as $JSQ(2)$, it will take time to balance the load since the new server is under-sampled.} Can $SCD$ be modified to resolve this issue, or should some other approach based on stochastic coordination be employed?

Clearly, distributed load balancing presents many interesting challenges, and a broad set of issues for study that can impact the efficiency of practical systems.

% \begin{itemize}
%     \item Practical algorithm idea: start with identical $\rate$'s (reduces to TWF) and adaptively update them based on periodic estimations from the servers.
%     \item More and more large scale computing systems -> more and more distributed aspects in practical systems have a significant impact. Growing need to understand the distributed principles of these systems.
%     \item future work: (1) not complete bipartite graphs; (2) non-identical jobs. E.g, different distributions of required work that might also depend on the server's identity.
% \end{itemize}

% you can include the acknowledgements in the source, but `anonymous' option will hide them
\section{Acknowledgements}

Guy Goren was partly supported by a grant from the Technion Hiroshi Fujiwara cyber security research center and the Israel cyber bureau, as well as by a Jacobs fellowship.
Yoram Moses is the Israel Pollak academic chair at the Technion. His work was supported in part by the Israel Science Foundation under grant 2061/19.

\bibliographystyle{plainurl}
\bibliography{references}

\appendix

%%%%%%%%%%%%%%%%%%%%%%%%%%%%%%%%%%%%%%%%%%%%%%%%%%%%%%%%%%%%%%%%%%
%%%%%%%%%%%%%%%%%%%%%%%%%%%%%%%%%%%%%%%%%%%%%%%%%%%%%%%%%%%%%%%%%%
%%%%%%%%%%%%%%%%%%%%%%%%%%%%%%%%%%%%%%%%%%%%%%%%%%%%%%%%%%%%%%%%%%
%%%%%%%%%%%%%%%%%%%%%%%%%%%%%%%%%%%%%%%%%%%%%%%%%%%%%%%%%%%%%%%%%%
%%%%%%%%%%%%%%%%%%%%%%%%%%%%%%%%%%%%%%%%%%%%%%%%%%%%%%%%%%%%%%%%%%

\newpage

\section{Pseudocode for Computing the IWL and IBA}
\label{app:alg:iba}

\Cref{alg:iba} computes the \iwl{} in a heterogeneous system.
Obtaining the \iba{} is immediate by using 
\begin{equation*}
    \frac{q_\serv + \receives_\serv}{\rate_\serv} = \max\left\{\frac{q_\serv}{\rate_\serv},\, \iwl \right\}.  \tag{\ref{eq:iba from iwl} revisited}
\end{equation*}
This computation of the \iwl{} and \iba{} is efficient. In particular, its complexity is in $O(n)$ if the servers are presorted by $\frac{q_\serv}{\rate_\serv}$.
Recall that we use `$a$' as a shorthand for the total \mbox{sum of arrivals:\, $a\,\triangleq \sum\limits_{\disp\in\setD}\!\!{a^{(\disp)}}$.}
%%=============================%%
\begin{algorithm}[ht]
  \DontPrintSemicolon
  \SetKwFunction{FMain}{ComputeIdealWorkLoad}
  \SetKwProg{Pn}{Function}{:}{\KwRet}
  \Pn{\FMain{$\setS$, $\setQ$, $\setDu$, $a$}}{
        $\setS_{\ge\wl} \gets \setS$;\,\, $\widthToFill \gets 0$;\,\, $\filled \gets a$\;
        
        $\curr \gets \arg\min_{\serv\in \setS_{\ge\wl}} \frac{q_\serv}{\rate_\serv}$\; 
        $\wl~\gets q_\curr/\rate_\curr$\;
        
        \While{$\filled > 0$}{
    %    \uIf{$\setS_{\ge\wl} = \emptyset$}
    %        {\Return $\wl + \frac{\filled}{\widthToFill}$}
    %     \Else{
            $\widthToFill \gets \widthToFill +\rate_{\curr}$\;
            $\setS_{\ge\wl} \gets \setS_{\ge\wl} {\setminus} \{ \curr \}$\;
            \If{$\setS_{\ge\wl} = \emptyset$}
            {\Return $\wl + \frac{\filled}{\widthToFill}$}
            $\curr \gets \arg\min_{\serv\in \setS_{\ge\wl}} \frac{q_\serv}{\rate_\serv}$\;
            $\Delta \gets \frac{q_{\curr}}{\rate_{\curr}} - \wl$\;
            \If{$\Delta \cdot \widthToFill \ge \filled $}{
           	        \Return $\wl + \frac{\filled}{\widthToFill}$\;
           	}
            $\filled\gets \filled - \Delta \cdot \widthToFill$\;
           	$\wl \gets \wl + \Delta$\;    
    %        }
        }
    }
    \caption{Computing the ideal workload.}
    \label{alg:iba}
\end{algorithm}
%%=============================%%

%%%%%%%%%%%%%%%%%%%%%%%%%%%%%%%%%%%%%%%%%%%%%%%%%%%%%%%%%%%%%%%%%%
%%%%%%%%%%%%%%%%%%%%%%%%%%%%%%%%%%%%%%%%%%%%%%%%%%%%%%%%%%%%%%%%%%

\section{Pseudocode for Computing the Probabilities with Optimal Complexity}
\label{app:alg:opt}

\Cref{alg:opt} computes the optimal probabilities in $O(n)$ time complexity given the ordering of the servers.
The main differences between the $O(n^2)$ \Cref{alg:n2} and the $O(n)$ \Cref{alg:opt} occur in:  
\begin{itemize}
    \item Calculating $\Lambda_0$ in line~7 of~\Cref{alg:n2} is substituted by lines~9-11 of~\Cref{alg:opt},
    
    \item The \textbf{for} loop in lines~8-11 of \Cref{alg:n2} tests whether the computed probabilities are all non-negative. This is replaced by the single $O(1)$ complexity test in lines~12-13 of \Cref{alg:opt},
    
    \item The summation in line~12 of \Cref{alg:n2} is replaced by the computation in lines~14-16 of \Cref{alg:opt}, which rely on \Cref{lem:2}.
\end{itemize}

%%=============================%%
\begin{algorithm}[ht]
  \DontPrintSemicolon
  \SetKwFunction{FMain}{ComputeProbabilities}
  \SetKwProg{Pn}{Function}{:}{\KwRet}
  \Pn{\FMain{$\setS$, $\setQ$, $\setDu$, $a$, \wl}}{
        $\setO \gets \emptyset$;\,\,$\val^* \gets\infty$\;
        $\Lambda_{0,n}\gets-2(a-1)$;\,\,$\Lambda_{0,d}\gets0$\;
        $p_\serv\gets0\,\forall\serv\in\setS$\;
        $v_1\gets0$;\,\,$v_2\gets0$\;
        \While{$\setS{\setminus}\setO \neq \emptyset$}
        {
            $\curr \gets \arg\min_{\serv\in \setS\setminus\setO} \frac{2q_\serv+1}{\rate_\serv}$\;
            $\setO \gets \setO \cup \curr$\;
            $\Lambda_{0,n} \gets \Lambda_{0,n} + 2(\rate_\curr \wl - q_\curr) - 1$\;
            $\Lambda_{0,d} \gets \Lambda_{0,d} + \rate_\curr$\;
            $\Lambda_0 \gets \frac{\Lambda_{0,n}}{\Lambda_{0,d}}$
            \tcp*[r]{according to \cref{eq:1 equality}}
            \If{$2\wl-\frac{2q_\curr+1}{\rate_\curr} < \Lambda_0$}
                    {Continue \tcp*[r]{the solution is infeasible; go to line~6}}
            $v_1 \gets v_1 + \frac{\rate_\curr}{4(a-1)}$\;  
            $v_2 \gets v_2 + \frac{(2(q_\curr - \rate_\curr\wl) + 1)^2 }{4\rate_\curr(a-1)}$\; 
            $\val \gets v_1 \Lambda_0^2 - v_2$ \tcp*[r]{according to \cref{eq:standard form}}
            \If{$\val < \val^*$} 
            {    
                $\val^* \gets \val$\;
                $\Lambda_0^*\gets \Lambda_0$\;
            }              
        }
        $P^* \gets \left\{ \max \{ 0, \,\,\frac{- 2(q_\serv - \rate_\serv\wl) - 1 - \rate_\serv \Lambda_0^*}{2(a-1)} \}  \right\}_{\serv\in\setS}$ \tcp*[r]{according to \cref{eq:positive p_s}}
        \Return $P^*$\;
    }
    \caption{Find probabilities in optimal time --- $O(n)$ given the servers' ordering.}
    \label{alg:opt}
\end{algorithm}
%%=============================%%

%%%%%%%%%%%%%%%%%%%%%%%%%%%%%%%%%%%%%%%%%%%%%%%%%%%%%%%%%%%%%%%%%%
%%%%%%%%%%%%%%%%%%%%%%%%%%%%%%%%%%%%%%%%%%%%%%%%%%%%%%%%%%%%%%%%%%

\section{Full Proofs}\label{app:proofs}

In this appendix, we give the complete proof of the two key Lemmas we employ in \Cref{sec:alg comp} to reduce our dispatching algorithm's computational complexity.

\subsection{Proof of Lemma \ref{lem:ordering}}

For clarity we restate the Lemma.

\orderingLemma*

\begin{proof}
        Let $\servT$ and $\servU$ satisfy the assumption that $\frac{2q_{\servT} + 1}{\rate_{\servT}} \ge \frac{2q_{\servU} + 1}{\rate_{\servU}}$, and let $P^*=\{p^*_1,\ldots, p^*_n\}$ be the optimal solution for \cref{eq:standard form}.
    Assume by contradiction that $p^*_{\servT}>0$ but $p^*_{\servU}=0$.
    We show that there exists a feasible solution $P$ that obtains a lower value of the objective function, thus contradicting the optimality of $P^*$.
    Specifically, we show that for the positive constant
    $$\const = \min \{\frac{\rate_{\servU} (2q_{\servT} + 1) - \rate_{\servT} (2q_{\servU} + 1) + 2\rate_{\servT} \rate_{\servU} p^*_{\servT}} {(a-1)(\rate_{\servT} + \rate_{\servU}) }, p^*_{\servT}\},$$ any $0< \epsilon < z$ and a different solution $P=\{p_1, \ldots, p_n\}$ with $p_{\servU}=\epsilon, p_{\servT}=p^*_{\servT}-\epsilon$, and $p_\serv=p^*_\serv$ for all other servers, it holds that $P$ is a better feasible solution than $P^*$.
    
    First, we show that $P$ is feasible.
    Since $P^*$ is feasible, it holds that $\sum\limits_{\serv\in\setS}p^*_\serv = 1$, and $0\le p^*_s \le 1$ for all $\serv\in\setS$.
    Accordingly, for $P$ it similarly holds that
    $$\sum\limits_{\serv\in\setS}p_\serv = \sum\limits_{\serv\in\setS\setminus\{\servT,\servU\}}p^*_\serv + p_{\servT} + p_{\servU} = \sum\limits_{\serv\in\setS}p^*_\serv -\epsilon + \epsilon = 1.$$
    The condition $0\le p_s \le 1$ trivially follows from $P^*$'s feasibility and the fact that $0< \epsilon < p^*_{\servT} \le 1$.
    
    Now we show that the new solution $P$ obtains in a lower objective function's value (i.e., by \cref{eq:standard form}) than the optimal solution $P^*\!$, leading to a contradiction.
    
    Denote \mbox{$\diff \triangleq f(P^*) - f(P)$}. Now, we show that $\diff>0$.
    By \cref{eq:standard form} we have
    
\begin{align*}
    \diff =&~
    \left((a-1)\sum_{\serv\in \setS} \frac{1}{\rate_\serv} \cdot (p^*_\serv)^2 +  \sum_{\serv\in \setS} \frac{2(q_\serv - \rate_\serv\wl) + 1}{\rate_\serv} \cdot p^*_\serv \right)
    \\&- \left( (a-1)\sum_{\serv\in \setS} \frac{1}{\rate_\serv} \cdot p_\serv^2 +  \sum_{\serv\in \setS} \frac{2(q_\serv - \rate_\serv\wl) + 1}{\rate_\serv} \cdot p_\serv \right)
    \\
    =&~
    \sum_{\serv\in \setS} \left( \frac{a-1}{\rate_\serv} (p^*_\serv)^2 +  \frac{2(q_\serv - \rate_\serv\wl) + 1}{\rate_\serv} p^*_\serv \right)
    - \sum_{\serv\in \setS} \left( \frac{a-1}{\rate_\serv}  p_\serv^2 + \frac{2(q_\serv - \rate_\serv\wl) + 1}{\rate_\serv} p_\serv \right).
\end{align*}
    Next, we split each summation term to a sum over $\setS\setminus\{\servT,\servU\}$ and a sum over $\{\servT,\servU\}$
\begin{align*}
    \diff =&~
    \sum_{\serv\in \setS\setminus\{\servT,\servU\}} \left( \frac{a-1}{\rate_\serv} (p^*_\serv)^2 +  \frac{2(q_\serv - \rate_\serv\wl) + 1}{\rate_\serv} p^*_\serv \right)
    +
    \sum_{\serv\in \{\servT,\servU\}} \left( \frac{a-1}{\rate_\serv} (p^*_\serv)^2 +  \frac{2(q_\serv - \rate_\serv\wl) + 1}{\rate_\serv} p^*_\serv \right)
    \\
    &- \sum_{\serv\in \setS\setminus\{\servT,\servU\}} \left( \frac{a-1}{\rate_\serv}  p_\serv^2 + \frac{2(q_\serv - \rate_\serv\wl) + 1}{\rate_\serv} p_\serv \right)
    -
    \sum_{\serv\in \{\servT,\servU\}} \left( \frac{a-1}{\rate_\serv}  p_\serv^2 + \frac{2(q_\serv - \rate_\serv\wl) + 1}{\rate_\serv} p_\serv \right),
\end{align*}
    and since $p_\serv = p_\serv^*$ for every server in $\setS\setminus\{\servT,\servU\}$, the sums over those sets cancel out.
    Thus, we obtain
\begin{align*}
    \diff =&~ 
    \sum_{\serv\in \{\servT,\servU\}} \left( \frac{a-1}{\rate_\serv} (p^*_\serv)^2 +  \frac{2(q_\serv - \rate_\serv\wl) + 1}{\rate_\serv} p^*_\serv \right)
    -
    \sum_{\serv\in \{\servT,\servU\}} \left( \frac{a-1}{\rate_\serv}  p_\serv^2 + \frac{2(q_\serv - \rate_\serv\wl) + 1}{\rate_\serv} p_\serv \right)
    \\
    =&~
     \sum_{\serv\in \{\servT,\servU\}} \left( \frac{a-1}{\rate_\serv} \left((p^*_\serv)^2 - p_\serv^2 \right) +  \frac{2(q_\serv - \rate_\serv\wl) + 1}{\rate_\serv} (p^*_\serv - p_\serv) \right)
    \\
    =&~ \left( \frac{a-1}{\rate_{\servT}} \left((p^*_{\servT})^2 - p_{\servT}^2 \right) +  \frac{2(q_{\servT} - \rate_{\servT}\wl) + 1}{\rate_{\servT}} (p^*_{\servT} - p_{\servT}) \right)
    \\&+
    \left( \frac{a-1}{\rate_{\servU}} \left((p^*_{\servU})^2 - p_{\servU}^2 \right) +  \frac{2(q_{\servU} - \rate_{\servU}\wl) + 1}{\rate_{\servU}} (p^*_{\servU} - p_{\servU}) \right).
\end{align*}
    We replace $(p^*_{\servT} - p_{\servT})$ and $(p_{\servU} - p^*_{\servU})$ by $\epsilon$. 
    Similarly, we replace $\left( (p^*_{\servT})^2 - p_{\servT}^2 \right)$ by $(2 p^*_{\servT}\epsilon - \epsilon^2)$, and $\left((p^*_{\servU})^2 - p_{\servU}^2 \right)$ by $\left(-\epsilon^2\right)$ and obtain
\begin{align*}
    \diff =&~ 
    \left( \frac{a-1}{\rate_{\servT}} \left(2 p^*_{\servT}\epsilon - \epsilon^2 \right) +  \frac{2(q_{\servT} - \rate_{\servT}\wl) + 1}{\rate_{\servT}} \cdot \epsilon \right)
    +
    \left( \frac{a-1}{\rate_{\servU}} \left( - \epsilon^2 \right) +  \frac{2(q_{\servU} - \rate_{\servU}\wl) + 1}{\rate_{\servU}} (- \epsilon) \right)
    \\
    =&~ \epsilon \left( \frac{2(q_{\servT} - \rate_{\servT}\wl) + 1}{\rate_{\servT}} - \frac{2(q_{\servU} - \rate_{\servU}\wl) + 1}{\rate_{\servU}}  +2p^*_{\servT}\right)
    -
    \epsilon^2 \left( \frac{a-1}{\rate_{\servT}} + \frac{a-1}{\rate_{\servU}} \right)
    \\
    =&~ \epsilon \underbrace{ \left( \frac{2q_{\servT} + 1}{\rate_{\servT}} - \frac{2q_{\servU} + 1}{\rate_{\servU}}  +2p^*_{\servT} \right) }_{x}
    -
    \epsilon^2 \underbrace{ \left( \frac{a-1}{\rate_{\servT}} + \frac{a-1}{\rate_{\servU}} \right) }_{y}.
\end{align*}
    Since $\frac{2q_{\servT} + 1}{\rate_{\servT}} \ge \frac{2q_{\servU} + 1}{\rate_{\servU}}$ and $p^*_{\servT}>0$, it must hold that $x>0$.
    It must also be true that $y>0$ since all $\rate$-s are positive and we consider only $a>1$.
    We therefore obtain that $0 < \frac{x}{y} = \frac{\rate_{\servU} (2q_{\servT} + 1) - \rate_{\servT} (2q_{\servU} + 1) + 2\rate_{\servT} \rate_{\servU} p^*_{\servT}} {(a-1)(\rate_{\servT} + \rate_{\servU}) }$. Hence, there exists $0 < \epsilon < \min \{ \frac{x}{y}, p^*_{\servT} \}$.
    Moreover, for such an $\epsilon$ it holds that 
    $$
    \diff = \epsilon x - \epsilon^2 y = \epsilon (x - \epsilon y) >  \epsilon (x - \frac{x}{y} y) = 0.
    $$
    
    Therefore, solutions exist that are both feasible and result in a lower objective function's value than that of the optimal solution --- a contradiction. This concludes the proof.
\end{proof}

%%%%%%%%%%%%%%%%%%%%%%%%%%%%%%%%%%%%%%%%%%%%%%%%%%%%%%%%%%%%%%%%%%

\subsection{Proof of Lemma \ref{lem:2}}

For clarity we restate the Lemma.

\lemTwo*

\begin{proof}
    By combining \cref{eq:positive p_s}, which expresses the positive optimal probabilities as a function of $\Lambda_0$, into the objective function of \cref{eq:standard form} we obtain
    \begin{equation*}
    \begin{aligned}
        f(P(\Lambda_0)) =&~  (a-1)\sum_{\serv \in \setSp} \frac{1}{\rate_\serv}  \left( \frac{- 2(q_\serv - \rate_\serv\wl) - 1 - \rate_\serv \Lambda_0}{2(a-1)} \right)^2
         \\ & + 
        \sum_{\serv \in \setSp} \frac{2(q_\serv - \rate_\serv\wl) + 1}{\rate_\serv} \cdot \left( \frac{- 2(q_\serv - \rate_\serv\wl) - 1 - \rate_\serv \Lambda_0}{2(a-1)} \right)\\
        =&  \sum_{\serv \in \setSp} \frac{1}{4\rate_\serv(a-1)} \cdot
            \left( (2(q_\serv - \rate_\serv\wl) + 1)^2 + (\rate_\serv \Lambda_0)^2 + 2(2(q_\serv - \rate_\serv\wl) + 1)\rate_\serv \Lambda_0 \right) \\
        & + \sum_{\serv \in \setSp} \frac{ -(2(q_\serv - \rate_\serv\wl) + 1)^2 - (2(q_\serv - \rate_\serv\wl) + 1) \rate_\serv \Lambda_0 }{2\rate_\serv(a-1)}\\
        =&  \sum_{\serv \in \setSp} \frac{(2(q_\serv - \rate_\serv\wl) + 1)^2}{4\rate_\serv(a-1)}
            + \Lambda_0 ^2\sum_{\serv \in \setSp} \frac{\rate_\serv}{4(a-1)} 
            + \Lambda_0 \sum_{\serv \in \setSp} \frac{2(2(q_\serv - \rate_\serv\wl) + 1)}{4(a-1)} \\
        & - \sum_{\serv \in \setSp} \frac{(2(q_\serv - \rate_\serv\wl) + 1)^2 }{2\rate_\serv(a-1)}
        - \Lambda_0 \sum_{\serv \in \setSp} \frac{(2(q_\serv - \rate_\serv\wl) + 1) }{2(a-1)} \\
        =&~ \Lambda_0^2\underbrace{\sum_{\serv \in \setSp} \frac{\rate_\serv}{4(a-1)}}_{v_1}
         - \underbrace{\sum_{\serv \in \setSp} \frac{(2(q_\serv - \rate_\serv\wl) + 1)^2 }{4\rate_\serv(a-1)}}_{v_2}
        .
    \end{aligned}
    \end{equation*}
This concludes the proof.
\end{proof}

%%%%%%%%%%%%%%%%%%%%%%%%%%%%%%%%%%%%%%%%%%%%%%%%%%%%%%%%%%%%%%%%%%
%%%%%%%%%%%%%%%%%%%%%%%%%%%%%%%%%%%%%%%%%%%%%%%%%%%%%%%%%%%%%%%%%%

\section{Strong stability}\label{app:stabiltiy}
%\shay{verify notations before submission}

In line with standard practice~\cite{vargaftik2020lsq,DISC2020,zhou2020asymptotically}, we make assumptions on the arrival and departure processes that make the system dynamics amenable to formal analysis. In particular, we assume, for all dispatchers $\disp\in\setD$,
\begin{equation}\label{eq:arrivals}
    \brac{a^{(\disp)}(t)}_{t=0}^{\infty} \text{ is an $i.i.d.$ process}, \quad \E[a^{(\disp)}(0)] = \lambda_\disp, \quad \E [(a^{(\disp)}(0))^2] = \sigma_\disp.
\end{equation}
Likewise, for all servers $\serv\in\setS$
\begin{equation}\label{eq:departures}
    \{c_\serv(t)\}_{t=0}^{\infty} \text{ is an $i.i.d.$ process}, \quad \E[c_\serv(0)] = \mu_\serv, \quad \E [(c_\serv(0))^2] = \varphi_\serv.
\end{equation}
Namely, for both the arrival and departure processes, we make the standard assumption that they are $i.i.d.$ and have a finite variance. We remark that the assumption that the arrival processes at the dispatchers are independent is made for ease of exposition and can be dropped by one skilled at the art at the cost of a more involved presentation (e.g., \cite{vargaftik2020lsq}). 

Intuitively, our goal is to prove that if, on average, the total arrivals at the system are below the total processing capacity of all servers, then the expected queue-lengths at the servers are bounded by a constant. We next present the required terms to formalize this intuition. 

\TT{Admissibility.} We assume the total expected arrival rate to the system is admissible. Formally, it means that we assume that there exists an $\epsilon > 0$ such that $\sum_{\serv} \mu_s - \sum_{\disp} \lambda_d = \epsilon$.

We prove that for any multi-dispatcher heterogeneous system with admissible arrivals our dispatching policy is strongly stable. 
Our proof follows similar lines to the strong stability proof in \cite{DISC2020}. 
The key difference is that we account for server heterogeneity. This makes the proof somewhat more involved.
We next formally define the well-established \emph{strong stability} criterion of interest.  
Strong stability is a strong form of stability for discrete-time queuing systems. Similarly to~\cite{vargaftik2020lsq,DISC2020,zhou2020asymptotically}, since we assume that the arrival and departure processes has a finite variance, this criterion also implies throughput optimality and other strong theoretical guarantees that may be of interest (see \cite{neely2010stability,neely2007optimal,georgiadis2006resource} for details).
\begin{definition}[Strong stability] A load balancing system is said to be strongly stable if for any admissible arrival rate it holds that 
$$\limsup_{T \to \infty} \frac{1}{T} \sum_{t=0}^{T-1} \sum_{\serv\in\setS} \E \Big[q_\serv(t)\Big] < \infty~.$$
\end{definition}
Now, we are ready to formalize the queue dynamics. Let $\receives_\serv(t) = \sum_{\disp\in\setD} \receives_\serv^{(\disp)}(t)$ be the total number of arrivals at server $\serv$ and round $t$.  Then, the recursion describing the queue dynamics of server $\serv$ over rounds is given by
\begin{equation}\label{eq:queue_dynamics}
q_\serv(t+1) = \max \{0, q_\serv(t) + \receives_\serv(t) - c_\serv(t)\}~. 
\end{equation}
Squaring both sides of \cref{eq:queue_dynamics}, rearranging, dividing by the server's processing capacity and omitting terms yields, 
\begin{equation}\label{eq:basic_nneq}
    \frac{1}{\mu_\serv}\bp{q_\serv(t+1)}^2 - \frac{1}{\mu_\serv}\bp{q_\serv(t)}^2 \le \frac{1}{\mu_\serv}\bp{\receives_\serv(t)}^2 + \frac{1}{\mu_\serv}\bp{c_\serv(t)}^2 - \frac{2q_\serv(t)}{\mu_\serv}\bp{c_\serv(t)-\receives_\serv(t)}~.    
\end{equation}
Summing \cref{eq:basic_nneq} over the servers yields
\begin{equation}\label{eq:sum_over_servers}
\begin{aligned}
  \sum_{\serv\in\setS} \frac{1}{\mu_\serv}\bp{q_\serv(t+1)}^2 & - \sum_{\serv\in\setS} \frac{1}{\mu_\serv}\bp{q_\serv(t)}^2 \le \\ &\sum_{\serv\in\setS} \frac{1}{\mu_\serv}\bp{\receives_\serv(t)}^2 + \sum_{\serv\in\setS} \frac{1}{\mu_\serv}\bp{c_\serv(t)}^2 - 2\sum_{\serv\in\setS} \frac{q_\serv(t)}{\mu_\serv}\bp{c_\serv(t)-\receives_\serv(t)}~.
\end{aligned}
\end{equation}
Denote $\mu_{tot}=\sum_{\serv\in\setS}\mu_\serv$. Now, we define the following useful quantity for each server $\serv$: $w_{\serv} = \frac{\mu_{\serv}}{\mu_{tot}}$.  
Next, for each $(\serv,\disp,k)$, let $I_\serv^{\disp,k}(t)$ be an indicator function that takes the value of 1 with probability $w_\serv$ and 0 otherwise such that 
$$\sum_{\serv\in\setS} I_\serv^{\disp,k}(t) = 1 \quad \forall \disp \in \setD, k \in [1,\ldots,a^{(\disp)}(t)]~.$$
We rewrite \cref{eq:sum_over_servers} and add and subtract the term 
$2 \sum_{\serv\in\setS}  \sum_{\disp\in\setD} \sum_{k=1}^{a^{(\disp)}(t)} I_\serv^{\disp,k}(t) q_\serv(t)$
from the right hand side of the equation. This yields, 
\begin{equation}\label{eq:splittable_abc}
\begin{split}
  & \sum_{\serv\in\setS} \frac{1}{\mu_\serv}\bp{q_\serv(t+1)}^2 - \sum_{\serv\in\setS} \frac{1}{\mu_\serv}\bp{q_\serv(t)}^2 \le \underbrace{\sum_{\serv\in\setS} \frac{1}{\mu_\serv}\bp{\receives_\serv(t)}^2 + \sum_{\serv\in\setS} \frac{1}{\mu_\serv}\bp{c_\serv(t)}^2}_{(a)} \cr
  &  - 2\underbrace{\sum_{\serv\in\setS} \frac{q_\serv(t)}{\mu_\serv}\bp{c_\serv(t)-\sum_{\disp\in\setD} \sum_{k=1}^{a^{(\disp)}(t)} I_\serv^{\disp,k}(t)}}_{(b)} \cr
  & + 2\underbrace{\sum_{\serv\in\setS} \frac{q_\serv(t)}{\mu_\serv}\bp{\sum_{\disp\in\setD} \receives_\serv^{(\disp)}(t)-\sum_{\disp\in\setD} \sum_{k=1}^{a^{(\disp)}(t)} I_\serv^{\disp,k}(t)}}_{(c)}~,
\end{split}
\end{equation}
where we also used $\receives_\serv(t) = \sum_{\disp\in\setD} \receives_\serv^{(\disp)}(t)$ in (c). 
Our goal now is to take the expectation of \cref{eq:splittable_abc}.
We start with analyzing Term (a) in \cref{eq:splittable_abc}. 
Denote $\mu_{\min} = \min\limits_{\serv\in\setS} \rate_\serv$. 
Taking expectation and using \cref{eq:arrivals} and \cref{eq:departures} we obtain
\begin{equation}\label{eq:splittable_a}
\begin{split}
&\mathbb{E} \bigg[\sum_{\serv\in\setS} \frac{1}{\mu_\serv}\bp{\receives_\serv(t)}^2 + \sum_{\serv\in\setS} \frac{1}{\mu_\serv}\bp{c_\serv(t)}^2\bigg]  \le 
\frac{1}{\mu_{\min}} \mathbb{E} \bigg[\Big(\sum_{\disp\in\setD} a^{(\disp)}(t) \Big)^2 \bigg] + \sum_{\serv\in\setS} \frac{\varphi_\serv}{\mu_\serv} \\ 
&= 
\frac{1}{\mu_{\min}} \bigg( \sum_{\disp\in\setD} \sigma_\disp + \sum_{\disp\in\setD}\sum_{\disp'\in\setD, \disp'\neq\disp} \lambda_\disp \lambda_{\disp'} \bigg) + \sum_{\serv\in\setS} \frac{\varphi_\serv}{\mu_\serv}
\triangleq C~.
\end{split}
\end{equation}
We now turn to analyze Term (b) in \cref{eq:splittable_abc}. Using the law of total expectation we obtain
\begin{equation}\label{eq:splittable_b}
\begin{split}
&\mathbb{E}\bigg[ \sum_{\serv\in\setS} \frac{q_\serv(t)}{\mu_\serv}\bp{c_\serv(t)-\sum_{\disp\in\setD} \sum_{k=1}^{a^{(\disp)}(t)} I_\serv^{\disp,k}(t)}\bigg] = \mathbb{E}\Bigg[ \mathbb{E}\bigg[ \sum_{\serv\in\setS} \frac{q_\serv(t)}{\mu_\serv}\bp{c_\serv(t)-\sum_{\disp\in\setD} \sum_{k=1}^{a^{(\disp)}(t)} I_\serv^{\disp,k}(t)}\bigg] \Bigg| \set{q_\serv(t)}_{\serv\in\setS} \Bigg] \\
& = \sum_{\serv\in\setS} \mathbb{E}\Bigg[ \frac{q_\serv(t)}{\mu_s} \bigg(\mu_\serv - w_\serv \sum_{\disp\in\setD} \lambda_\disp \bigg) \Bigg] = \sum_{\serv\in\setS} \mathbb{E}\Bigg[ \frac{q_\serv(t)}{\mu_{tot}} \bigg(\mu_{tot} - \sum_{\disp\in\setD} \lambda_\disp \bigg) \Bigg] = \frac{\epsilon}{\mu_{tot}} \sum_{\serv\in\setS} \mathbb{E}\big[ q_\serv(t) \big]~,
\end{split}
\end{equation}
where we used \cref{eq:arrivals}, \cref{eq:departures}, the definition of $I_\serv^{\disp,k}(t)$ and the admissibility of the system. We also used the fact that $a^{(\disp)}(t)$ and $I_\serv^{\disp,k}(t)$ are independent which allowed us to employ Wald's identity. We turn to analyze term (c) in \cref{eq:splittable_abc}. 
Observer that
\begin{equation}\label{eq:switch_sum_order}
\sum_{\serv\in\setS} \frac{q_\serv(t)}{\mu_\serv}\bp{\sum_{\disp\in\setD} \receives_\serv^{(\disp)}(t)-\sum_{\disp\in\setD} \sum_{k=1}^{a^{(\disp)}(t)} I_\serv^{\disp,k}(t)} = \sum_{\disp\in\setD} \Bigg( \sum_{\serv\in\setS} \frac{q_\serv(t)}{\mu_\serv}\bp{\receives_\serv^{(\disp)}(t)-\sum_{k=1}^{a^{(\disp)}(t)} I_\serv^{\disp,k}(t)} \Bigg)~.
\end{equation}
This means that we can consider each dispatcher separately. In particular, for each dispatcher, we will apply the following Lemma.
\begin{lemma}\label{lemma:prepate_majorization}
Consider dispatcher $\disp\in\setD$. For any $a^{(d)}(t) > 0$, let $\set{p_\serv^{(\disp)}(t)}_{\serv\in\setS}$ be the optimal probabilities computed by dispatcher $\disp$ for round $t$ (i.e., probabilities that respect \cref{eq:standard form} for the computed $\wl$). Then, for any $p_\serv^{(\disp)}(t), p_{\serv'}^{(\disp)}(t) > 0$, it holds that 
\begin{equation}\label{eq:todo1}
\begin{split}
\frac{p_\serv^{(\disp)}(t)}{\mu_\serv} \le \frac{p_{\serv'}^{(\disp)}(t)}{\mu_{\serv'}} \rightarrow  \frac{q_\serv(t)}{\mu_\serv} + \frac{a^{(\disp)}(t)}{\mu_{min}} \ge \frac{q_\serv(t) + a^{(\disp)}(t)}{\mu_\serv} \ge \frac{q_{\serv'}(t)}{\mu_{\serv'}}~.
\end{split}
\end{equation}
\end{lemma}
\begin{proof}
    See Appendix \ref{app:prepate_majorization}
\end{proof}
With this result at hand, we produce the following construction. Let $p_{\serv_{(1)}}^{(\disp)}(t), p_{\serv_{(2)}}^{(\disp)}(t), \ldots, p_{\serv_{(n)}}^{(\disp)}(t)$ be the optimal probabilities computed by dispatcher $\disp$ ordered in a decreasing order of their value divided by the corresponding server rate, i.e., $\frac{p_{\serv}^{(\disp)}(t)}{\mu_\serv}$. Then, according to Lemma \ref{lemma:prepate_majorization}, for any two servers $\serv_{(i)}, \serv_{(i+1)}$ with positive probabilities, we have that
\begin{equation}\label{eq:todo2}
\begin{split}
\frac{q_{\serv_{(i+1)}}(t)}{\mu_{\serv_{(i+1)}}}  + i \cdot \frac{a^{(\disp)}(t)}{\mu_{min}}\ge \frac{q_{\serv_{(i)}}(t)}{\mu_{\serv_{(i)}}}  + (i-1) \cdot \frac{a^{(\disp)}(t)}{\mu_{min}}~.
\end{split}
\end{equation}
Now, for dispatcher $\disp\in\setD$ consider
\begin{equation}\label{eq:todo3}
\begin{split}
\sum_{{\serv_{(i)}}\in\setS} & \frac{q_{\serv_{(i)}}(t)}{\mu_{\serv_{(i)}}}\bp{\receives_{\serv_{(i)}}^{(\disp)}(t)- \sum_{k=1}^{a^{(\disp)}(t)} I_{\serv_{(i)}}^{\disp,k}(t)} =\\& \underbrace{\sum_{{\serv_{(i)}}\in\setS} \bp{\frac{q_{\serv_{(i)}}(t)}{\mu_{\serv_{(i)}}} + (i-1) \cdot \frac{a^{(\disp)}(t)}{\mu_{min}}} \bp{\receives_{\serv_{(i)}}^{(\disp)}(t)- \sum_{k=1}^{a^{(\disp)}(t)} I_{\serv_{(i)}}^{\disp,k}(t)}}_{(i)} \\ 
& - \underbrace{\sum_{{\serv_{(i)}}\in\setS} \bp{(i-1) \cdot \frac{a^{(\disp)}(t)}{\mu_{min}}} \bp{\receives_{\serv_{(i)}}^{(\disp)}(t)- \sum_{k=1}^{a^{(\disp)}(t)} I_{\serv_{(i)}}^{\disp,k}(t)}}_{(ii)}.
\end{split}
\end{equation}
For (i), we use Lemma \ref{lemma:majorization} which relies on a \emph{majorization} argument similarly to \cite{vargaftik2020lsq,zhou2020asymptotically}.
\begin{lemma}\label{lemma:majorization}
\begin{equation}\label{eq:todo4}
\begin{split} 
\mathbb{E} \Bigg[ \sum_{{\serv_{(i)}}\in\setS} \bp{\frac{q_{\serv_{(i)}}(t)}{\mu_{\serv_{(i)}}} + (i-1) \cdot \frac{a^{(\disp)}(t)}{\mu_{min}}} \bp{\receives_{\serv_{(i)}}^{(\disp)}(t)- \sum_{k=1}^{a^{(\disp)}(t)} I_{\serv_{(i)}}^{\disp,k}(t)} \Bigg]  \le 0~.
\end{split}
\end{equation}
\end{lemma}
\begin{proof}
    See Appendix \ref{app:majorization}. 
\end{proof}

\noindent We continue to analyze (ii). By \cref{eq:arrivals}, it holds that 
\begin{equation}\label{eq:todo5}
\begin{split} 
\mathbb{E} \Bigg[\sum_{{\serv_{(i)}}\in\setS} \bp{(i-1) \cdot \frac{a^{(\disp)}(t)}{\mu_{min}}} & \bp{\receives_{\serv_{(i)}}^{(\disp)}(t)- \sum_{k=1}^{a^{(\disp)}(t)} I_{\serv_{(i)}}^{\disp,k}(t)}\Bigg] \le \\& \mathbb{E} \Bigg[\sum_{{\serv_{(i)}}\in\setS} \bp{(i-1) \cdot \frac{(a^{(\disp)}(t))^2}{\mu_{min}}}\Bigg] = \frac{\sigma_{d}(n^2-n)}{2\mu_{min}}~.
\end{split}
\end{equation}
By taking the expectation of Term (c) and using \cref{eq:switch_sum_order}, Lemma \ref{lemma:majorization} and \cref{eq:todo5}, we obtain 
\begin{equation}\label{eq:sum_over_d_term_c}
\mathbb{E} \Bigg[ \sum_{\serv\in\setS} \frac{q_\serv(t)}{\mu_\serv}\bp{\sum_{\disp\in\setD} \receives_\serv^{(\disp)}(t)-\sum_{\disp\in\setD} \sum_{k=1}^{a^{(\disp)}(t)} I_\serv^{\disp,k}(t)} \Bigg] \le \sum_{\disp\in\setD} \frac{\sigma_{d}(n^2-n)}{2\mu_{min}} \triangleq D~.
\end{equation}

Finally, by applying \cref{eq:splittable_a}, \cref{eq:splittable_b} and \cref{eq:sum_over_d_term_c}, we can take the expectation of the right-hand side of \cref{eq:splittable_abc}. This yields
\begin{equation}\label{eq:todo6}
\begin{split}
  & \mathbb{E} \bigg[\sum_{\serv\in\setS} \frac{1}{\mu_\serv}\bp{q_\serv(t+1)}^2 \bigg] - \mathbb{E} \bigg[\sum_{\serv\in\setS} \frac{1}{\mu_\serv}\bp{q_\serv(t)}^2 \bigg] \le C - \frac{2\epsilon}{\mu_{tot}} \sum_{\serv\in\setS} \mathbb{E}\big[ q_\serv(t) \big] + 2D~.
\end{split}
\end{equation}
Summing \cref{eq:todo6} over rounds $0,\ldots,T{-}1$, multiplying by $\frac{\mu_{tot}}{2\epsilon T}$ and rearranging yields,
\begin{equation}\label{eq:upsplittable_abc_expectation_simple}
\frac{1}{T} \sum_{t=0}^{T-1} \sum_{\serv\in\setS} \E \Big[q_\serv(t)\Big] \le \frac{(C+2D)\mu_{tot}}{2\epsilon} + \frac{\mu_{tot}}{2\epsilon T} \E \sum_{\serv\in\setS} \frac{1}{\mu_\serv}\bp{q_\serv(0)}^2~,
\end{equation}
where we omitted the non-positive term $\E \Big[-\sum_{\serv\in\setS} \frac{1}{\mu_\serv}\bp{q_\serv(T)}^2\Big]$ as a results of the telescopic series at the left hand side of \cref{eq:todo6}. Taking limits of \cref{eq:upsplittable_abc_expectation_simple} and making the standard assumption that the system is initialized with bounded queue-lengths, \ie \mbox{$\E \Big[\sum_{\serv\in\setS} \bp{q_\serv(0)}^2\Big] < \infty$} yields, 
\begin{equation}\label{eq:unsplittable_final}
\begin{split}
    &  \limsup_{T \to \infty} \frac{1}{T} \sum_{t=0}^{T-1} \sum_{\serv\in\setS} \E \Big[q_\serv(t)\Big] \le \frac{(C+2D)\mu_{tot}}{2\epsilon} < \infty.
\end{split}
\end{equation}
This concludes the proof.

Note that the proof does not rely on the specific method to estimate the total arrivals (which, in turn, affects the calculated dispatching probabilities). We only rely on the fact that the estimation respects $1 \le a_{est,\disp} \le \infty$ and the way the probabilities are calculated as a result of the estimation. This allows one the freedom to design even better estimation techniques than \cref{eq:estimating a} as we briefly mention in the main text of the paper.

\subsection{Proof of Lemma \ref{lemma:prepate_majorization}}\label{app:prepate_majorization}

Assume by the way of contradiction that $\frac{p_\serv^{(\disp)}(t)}{\mu_\serv} \le \frac{p_{\serv'}^{(\disp)}(t)}{\mu_{\serv'}}$ but $\frac{q_\serv(t) + a^{(\disp)}(t)}{\mu_\serv} < \frac{q_{\serv'}(t)}{\mu_{\serv'}}$. 
For ease of exposition, we next omit the dispatcher superscript $\disp$ and the round index $t$.
Consider an alternative solution where all other probabilities are identical except that we set $\tilde{p}_\serv = p_\serv + p_{\serv'}$ and $\tilde{p}_{\serv'} = 0$. By the definition of the error function given in \cref{eq:standard form}, we have that the difference in the error is
\begin{equation*}
\begin{aligned}
diff = &(a-1) \frac{(p_\serv+p_{\serv'})^2}{\mu_\serv} + 2 (p_\serv+p_{\serv'}) \bp{  \frac{q_\serv}{\mu_\serv} - \wl + \frac{0.5}{\mu_\serv}}  \\&    -\Bigg( (a-1) \frac{p_\serv^2}{\mu_\serv} + 2 p_\serv \bp{  \frac{q_\serv}{\mu_\serv} - \wl + \frac{0.5}{\mu_\serv}} + (a-1) \frac{p_{\serv'}^2}{\mu_{\serv'}} + 2 p_{\serv'} \bp{  \frac{q_{\serv'}}{\mu_{\serv'}} - \wl + \frac{0.5}{\mu_{\serv'}}} \Bigg).  %     
\end{aligned}
\end{equation*}
Simplifying yields,
$$diff = (a-1)\bp{ \frac{p_{\serv'}^2 + 2 p_\serv p_{\serv'}}{\mu_\serv}  - \frac{p_{\serv'}^2}{\mu_{\serv'}}}  + 2 p_{\serv'} \bp{  \frac{q_\serv}{\mu_\serv} + \frac{0.5}{\mu_\serv}} - 2 p_{\serv'} \bp{  \frac{q_{\serv'}}{\mu_{\serv'}} + \frac{0.5}{\mu_{\serv'}}}.$$
We now split into two cases. 

\TT{Case 1: $\mu_\serv \ge \mu_{\serv'}$.} In this case
$$diff \le (a-1)\bp{ \frac{2 p_\serv p_{\serv'}}{\mu_\serv}}  + 2 p_{\serv'} \bp{  \frac{q_\serv}{\mu_\serv}} - 2 p_{\serv'} \bp{  \frac{q_{\serv'}}{\mu_{\serv'}}} < (a-1)\bp{ \frac{2 p_\serv p_{\serv'}}{\mu_\serv}} - a\frac{2 p_{\serv'}}{\mu_\serv} < 0.$$

\TT{Case 2: $\mu_\serv < \mu_{\serv'}$.} In this case
$$diff \le (a-1)\bp{ \frac{2 p_\serv p_{\serv'}}{\mu_{\serv'}}}  + 2 p_{\serv'} \bp{  \frac{q_\serv}{\mu_\serv}} - 2 p_{\serv'} \bp{  \frac{q_{\serv'}}{\mu_{\serv'}}} + \frac{p_{\serv'}}{\mu_\serv} < (a-1)\bp{ \frac{2 p_\serv p_{\serv'}}{\mu_{\serv'}}} - a\frac{2 p_{\serv'}}{\mu_\serv}  + \frac{p_{\serv'}}{\mu_\serv} < 0.$$
In both cases the new solution has a lower error. This is a contradiction to the assumption that $\frac{q_\serv(t) + a^{(\disp)}(t)}{\mu_\serv} < \frac{q_{\serv'}(t)}{\mu_{\serv'}}$. This concludes the proof.

\subsection{Proof of Lemma \ref{lemma:majorization}}\label{app:majorization}

For any $\serv\in\setS$, we have that 
$$ \mathbb{E} \Bigg[ \receives_{\serv_{(i)}}^{(\disp)}(t) \,\bigg|\,  a^{(\disp)}(t) , \set{q_\serv(t)}_{\serv\in\setS} \Bigg] =  p_{\serv_{(i)}}^{(\disp)}(t) \cdot a^{(d)}(t),$$
and,
$$ \mathbb{E} \Bigg[ \sum_{k=1}^{a^{(\disp)}(t)} I_{\serv_{(i)}}^{\disp,k}(t) \,\bigg|\,  a^{(\disp)}(t) , \set{q_\serv(t)}_{\serv\in\setS} \Bigg] =  \frac{\mu_{\serv_{(i)}}}{\mu_{tot}}\cdot a^{(d)}(t).$$
This means that
$$ (i)\quad p_{\serv_{(i)}}^{(\disp)}(t) \cdot a^{(d)}(t) = \mu_{\serv_{(i)}} \cdot a^{(d)}(t) \cdot \frac{p_{\serv_{(i)}}^{(\disp)}(t)}{\mu_{\serv_{(i)}}}, \quad\quad\quad  (ii)\quad \frac{\mu_{\serv_{(i)}}}{\mu_{tot}}\cdot a^{(d)}(t) = \mu_{\serv_{(i)}} \cdot a^{(d)}(t) \cdot \frac{1}{\mu_{tot}}.$$
Observe that the term $\mu_{\serv_{(i)}} \cdot a^{(d)}(t)$ is identical in both $(i)$ and $(ii)$, and recall that $\frac{p_{\serv_{(i)}}^{(\disp)}(t)}{\mu_{\serv_{(i)}}}$ is monotonically \emph{non-increasing} in $i$.
This allows us to obtain the following \emph{majorization}~\cite{enwiki:993822990} Lemma.
\begin{lemma}\label{lemma:formalise_majorization_1}
Let $k \in \set{0,1,\ldots,n-1}$. Then,
$$\sum_{i=n}^{n-k} p_{\serv_{(i)}}^{(\disp)}(t) \cdot a^{(d)}(t) \le \sum_{i=n}^{n-k} \frac{\mu_{\serv_{(i)}}}{\mu_{tot}}\cdot a^{(d)}(t) \quad \forall \, k~,$$ 
where an equality is obtained for $k=n-1$.
\end{lemma}
\begin{proof}
    See Appendix \ref{app:formalise_majorization_1}.
\end{proof}
Finally, since $\Big( \frac{q_{\serv_{(i)}}(t)}{\mu_{\serv_{(i)}}} + (i-1) \cdot \frac{a^{(\disp)}(t)}{\mu_{min}} \Big)$ is a monotonically \emph{non-decreasing} in $i$, we obtain the following Lemma.
\begin{lemma}\label{lemma:formalise_majorization_2}
$$\sum_{i=n}^{n-k} \bigg( \frac{q_{\serv_{(i)}}(t)}{\mu_{\serv_{(i)}}} + (i-1) \cdot \frac{a^{(\disp)}(t)}{\mu_{min}} \bigg) p_{\serv_{(i)}}^{(\disp)}(t) \cdot a^{(d)}(t) \le \sum_{i=n}^{n-k} \bigg( \frac{q_{\serv_{(i)}}(t)}{\mu_{\serv_{(i)}}} + (i-1) \cdot \frac{a^{(\disp)}(t)}{\mu_{min}} \bigg) \frac{\mu_{\serv_{(i)}}}{\mu_{tot}}\cdot a^{(d)}(t) \quad \forall \, k~.$$ 
\end{lemma}
\begin{proof}
    See Appendix \ref{app:formalise_majorization_2}.
\end{proof}
In particular,  Lemma \ref{lemma:formalise_majorization_2} holds for $k=n-1$. 
Therefore, using the Law of total expectation and Lemma \ref{lemma:formalise_majorization_2} we obtain
\begin{equation}
\begin{split}
&\mathbb{E} \Bigg[ \mathbb{E} \bigg[ \sum_{{\serv_{(i)}}\in\setS} \bp{\frac{q_{\serv_{(i)}}(t)}{\mu_{\serv_{(i)}}} + (i-1) \cdot \frac{a^{(\disp)}(t)}{\mu_{min}}} \bp{\receives_{\serv_{(i)}}^{(\disp)}(t)- \sum_{k=1}^{a^{(\disp)}(t)} I_{\serv_{(i)}}^{\disp,k}(t)} \,\bigg|\,  a^{(\disp)}(t) , \set{q_\serv(t)}_{\serv\in\setS} \bigg] \Bigg] = \\
&\mathbb{E} \Bigg[ \sum_{{\serv_{(i)}}\in\setS} \bp{\frac{q_{\serv_{(i)}}(t)}{\mu_{\serv_{(i)}}} + (i-1) \cdot \frac{a^{(\disp)}(t)}{\mu_{min}}} \bp{p_{\serv_{(i)}}^{(\disp)}(t) \cdot a^{(d)}(t)- \frac{\mu_{\serv_{(i)}}}{\mu_{tot}}\cdot a^{(d)}(t)}\Bigg] \le 0.
\end{split}
\end{equation}
This concludes the proof. 

\subsection{Proof of Lemma \ref{lemma:formalise_majorization_1} }\label{app:formalise_majorization_1}
By the way of contradiction, consider the smallest $k \in \set{0,1,\ldots,n-1}$ such that
$$\sum_{i=n}^{n-k} \mu_{\serv_{(i)}} \cdot a^{(d)}(t) \cdot \frac{p_{\serv_{(i)}}^{(\disp)}(t)}{\mu_{\serv_{(i)}}} > \sum_{i=n}^{n-k} \mu_{\serv_{(i)}} \cdot a^{(d)}(t) \cdot \frac{1}{\mu_{tot}}~.$$
Dividing both terms by $a^{(d)}(t)$ yields
$$\sum_{i=n}^{n-k} \mu_{\serv_{(i)}} \frac{p_{\serv_{(i)}}^{(\disp)}(t)}{\mu_{\serv_{(i)}}} > \sum_{i=n}^{n-k} \mu_{\serv_{(i)}} \cdot \frac{1}{\mu_{tot}}~.$$
Since $\frac{p_{\serv_{(i)}}^{(\disp)}(t)}{\mu_{\serv_{(i)}}}$ is monotonically \emph{non-increasing} sequence in $i$, we must have that if $k' \ge k$ then $\frac{p_{\serv_{(n-k')}}^{(\disp)}(t)}{\mu_{\serv_{(n-k')}}} > \frac{1}{\mu_{tot}}$. This means that it must hold that 
$$\sum_{i=n}^{1} \mu_{\serv_{(i)}} \frac{p_{\serv_{(i)}}^{(\disp)}(t)}{\mu_{\serv_{(i)}}} > \sum_{i=n}^{1} \mu_{\serv_{(i)}} \cdot \frac{1}{\mu_{tot}} = 1~.$$
 This is a contradiction to the fact that $\sum_{i=n}^{1} \mu_{\serv_{(i)}} \cdot \frac{1}{\mu_{tot}} = \sum_{i=1}^{n} \mu_{\serv_{(i)}} \cdot\frac{p_{\serv_{(i)}}^{(\disp)}(t)}{\mu_{\serv_{(i)}}} = \sum_{i=1}^{n} p_{\serv_{(i)}}^{(\disp)}(t) = 1$ (that is, both are probability vectors and thus their sum must be equal to 1).

\subsection{Proof of Lemma \ref{lemma:formalise_majorization_2} }\label{app:formalise_majorization_2}

Denote $\eta_{(i)} = \bigg( \frac{q_{\serv_{(i)}}(t)}{\mu_{\serv_{(i)}}} + (i-1) \cdot \frac{a^{(\disp)}(t)}{\mu_{min}} \bigg)$ and recall that $\eta_{(i)}$ is a monotonically \emph{non-decreasing} in $i$ and that $\frac{p_{\serv_{(i)}}^{(\disp)}(t)}{\mu_{\serv_{(i)}}}$ is monotonically \emph{non-increasing} in $i$. 
Consider the smallest index $j$ such that
$p_{\serv_{(j)}}^{(\disp)}(t) \le \frac{\mu_{\serv_{(j)}}}{\mu_{tot}}$. 
If $j=1$ than by Lemma \ref{lemma:formalise_majorization_1} and since $\sum_{i=1}^{n} p_{\serv_{(i)}}^{(\disp)}(t) = 1$, it must hold that $p_{\serv_{(j)}}^{(\disp)}(t) = \frac{\mu_{\serv_{(j)}}}{\mu_{tot}} \,\, \forall \, i$ and we are done. Now, assume $j>1$.
Consider the sums
$$\sum_{i=j-1}^{1} \eta_{(i)} \cdot ( p_{\serv_{(i)}}^{(\disp)}(t) - \frac{\mu_{\serv_{(i)}}}{\mu_{tot}} ) \le \max \set{\eta_{(i)}}_{i \in \set{1,\ldots,j-1}} \sum_{i=1}^{j-1}  ( p_{\serv_{(i)}}^{(\disp)}(t) - \frac{\mu_{\serv_{(i)}}}{\mu_{tot}} )~,$$
and
$$\sum_{i=n}^{j} \eta_{(i)} \cdot ( p_{\serv_{(i)}}^{(\disp)}(t) - \frac{\mu_{\serv_{(i)}}}{\mu_{tot}} ) \le \min \set{\eta_{(i)}}_{i \in \set{j,\ldots,n}} \sum_{i=n}^{j}  ( p_{\serv_{(i)}}^{(\disp)}(t) - \frac{\mu_{\serv_{(i)}}}{\mu_{tot}} )~.$$
Notice that 
$$ \max \set{\eta_{(i)}}_{i \in \set{1,\ldots,j-1}} \le \min \set{\eta_{(i)}}_{i \in \set{j,\ldots,n}}~,$$
and 
$$\sum_{i=1}^{j-1}  ( p_{\serv_{(i)}}^{(\disp)}(t) - \frac{\mu_{\serv_{(i)}}}{\mu_{tot}} ) =  - \sum_{i=n}^{j}  ( p_{\serv_{(i)}}^{(\disp)}(t) - \frac{\mu_{\serv_{(i)}}}{\mu_{tot}} )~.$$
Therefore, we obtain

\begin{equation}
\begin{split}
\sum_{i=n}^{1} & \eta_{(i)} \cdot ( p_{\serv_{(i)}}^{(\disp)}(t) - \frac{\mu_{\serv_{(i)}}}{\mu_{tot}} ) = \sum_{i=j-1}^{1} \eta_{(i)} \cdot ( p_{\serv_{(i)}}^{(\disp)}(t) - \frac{\mu_{\serv_{(i)}}}{\mu_{tot}} ) + \sum_{i=n}^{j} \eta_{(i)} \cdot ( p_{\serv_{(i)}}^{(\disp)}(t) - \frac{\mu_{\serv_{(i)}}}{\mu_{tot}} )\\
& \le \max \set{\eta_{(i)}}_{i \in \set{1,\ldots,j-1}} \sum_{i=1}^{j-1}  ( p_{\serv_{(i)}}^{(\disp)}(t) - \frac{\mu_{\serv_{(i)}}}{\mu_{tot}} ) + \min \set{\eta_{(i)}}_{i \in \set{j,\ldots,n}} \sum_{i=n}^{j}  ( p_{\serv_{(i)}}^{(\disp)}(t) - \frac{\mu_{\serv_{(i)}}}{\mu_{tot}} ) \\
& = \sum_{i=1}^{j-1}  ( p_{\serv_{(i)}}^{(\disp)}(t) - \frac{\mu_{\serv_{(i)}}}{\mu_{tot}} ) \bp{ \max \set{\eta_{(i)}}_{i \in \set{1,\ldots,j-1}} - \min \set{\eta_{(i)}}_{i \in \set{j,\ldots,n}} } \le 0
\end{split}
\end{equation}
This concludes the proof.\newpage
%%%%%%%%%%%%%%%%%%%%%%%%%%%%%%%%%%%
%%%%%%%%%%%%%%%%%%%%%%%%%%%%%%%%%%%
%%%%%%%%%%%%%%%%%%%%%%%%%%%%%%%%%%%
%%%%%%%%%%%%%%%%%%%%%%%%%%%%%%%%%%%

\begin{figure}[t]
    \centering
    \begin{subfigure}{0.8\textwidth}
      \centering
      \caption{Average response time.}
      %\vspace{-2mm}
    \includegraphics[width=\linewidth]{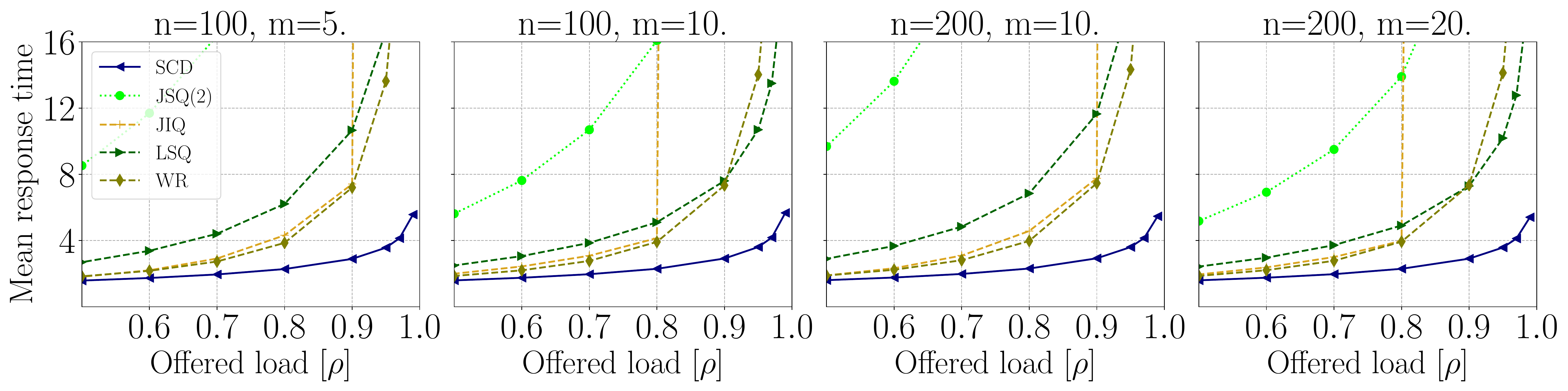}
     \label{fig:app:evaluation:moderate:loadsweep}
    \end{subfigure}
    \begin{subfigure}{0.8\textwidth}
      \centering
      \caption{Response time delay tail.}
      %\vspace{-2mm}
    \includegraphics[width=\linewidth]{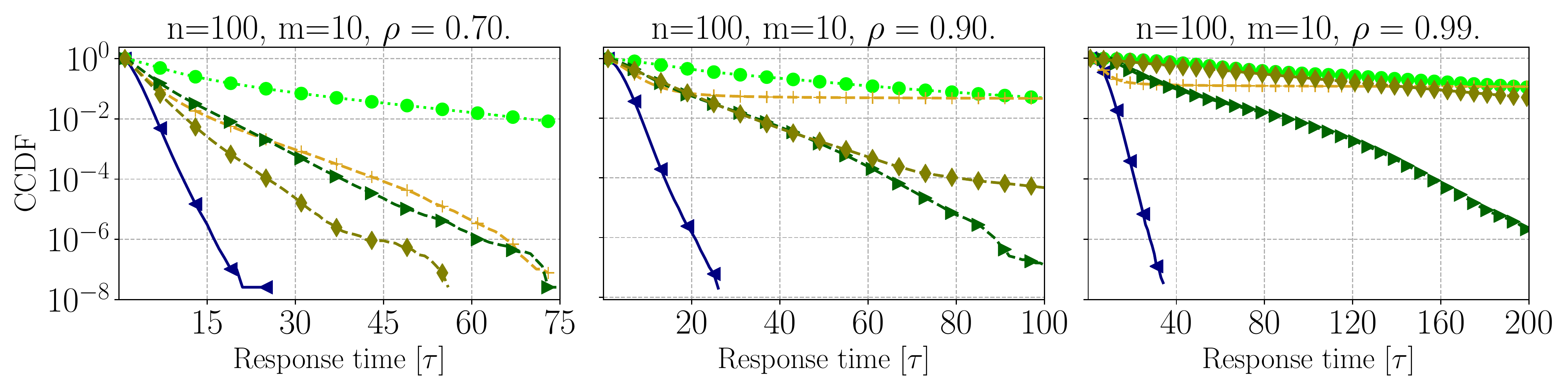}  
     \label{fig:app:evaluation:moderate:delaytail}
    \end{subfigure}
    \caption{ Complementary results to Figure \ref{fig:evaluation:moderate}. Comparing $SCD$ versus $JSQ(2)$, $JIQ$ $LSQ$ and $WR$ where $\mu_\serv \sim U[1,10]$. }
     %\vspace{-3mm}
    \label{fig:app:evaluation:moderate}
\end{figure}

\begin{figure}[ht]
    \centering
    \begin{subfigure}{0.8\textwidth}
      \centering
      \caption{Average response time.}
      %\vspace{-2mm}
    \includegraphics[width=\linewidth]{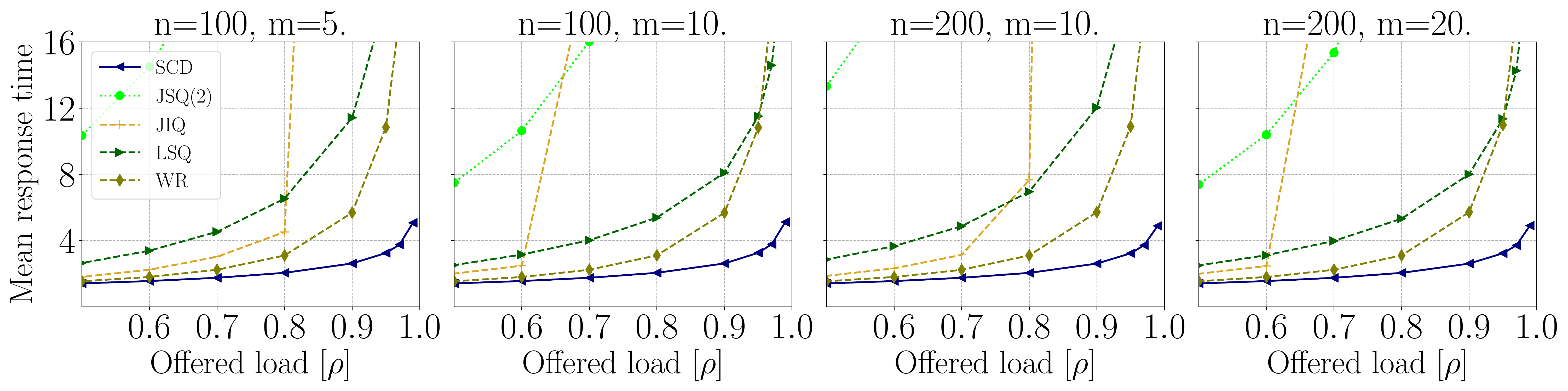}
     \label{fig:app:evaluation:high:loadsweep}
    \end{subfigure}
    \begin{subfigure}{0.8\textwidth}
      \centering
      \caption{Response time delay tail.}
      %\vspace{-2mm}
    \includegraphics[width=\linewidth]{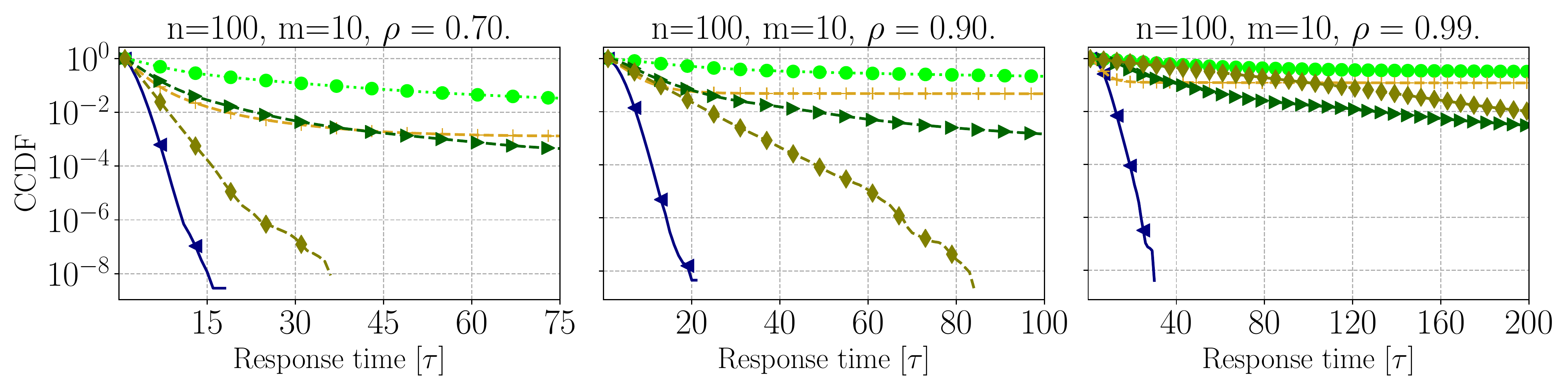}  
     \label{fig:app:evaluation:high:delaytail}
    \end{subfigure}
    \caption{ Complementary results to Figure \ref{fig:evaluation:high}. Comparing $SCD$ versus $JSQ(2)$, $JIQ$ $LSQ$ and $WR$ where $\mu_\serv \sim U[1,100]$. }
     %\vspace{-3mm}
    \label{fig:app:evaluation:high}
\end{figure}

\section{Additional simulation results}\label{app:exatra_results}

We next overview additional simulation results omitted from the main text for clarity and interest of space. 

\subsection{Response time}\label{app:Response time}

We show complementary results to Figures \ref{fig:evaluation:moderate} and \ref{fig:evaluation:high} (Section \ref{sec:eval:Response time}) comparing $SCD$ to $JSQ(2)$, $JIQ$ $LSQ$ and weighted random ($WR$)\footnote{In $WR$, each request is send to server $\serv$ with probability $\frac{\mu_\serv}{\mu_{tot}}$.}.
It is evident how $SCD$ significantly outperforms all these techniques across all systems, metrics, and offered loads. Indeed, these techniques are less competitive than the six presented in the main text. This is because $JSQ(2)$, $JIQ$, and $LSQ$ do not account for server heterogeneity, and $WR$ do account for server heterogeneity but ignores queue length information.

%%% Execution run-time
\begin{figure}[t]
    \centering
    \includegraphics[width=\linewidth]{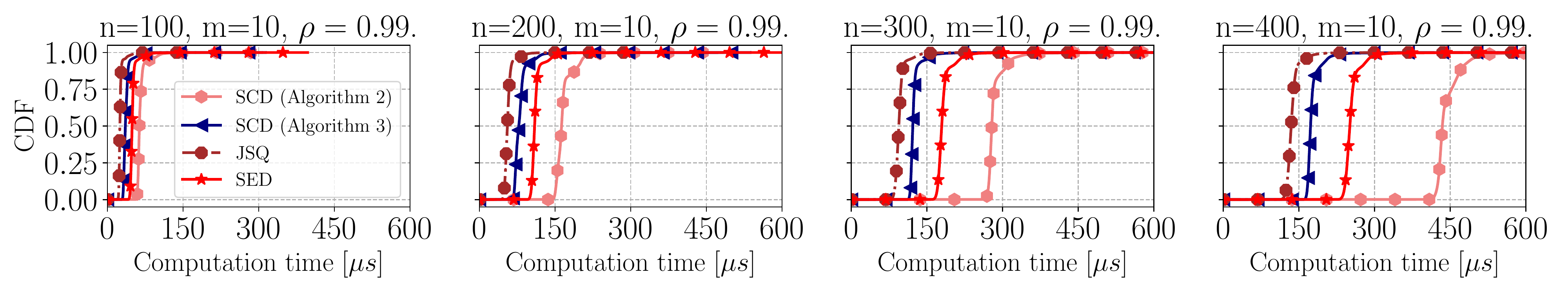}  
    \caption{ Evaluation of execution run-time over systems with an increasing number of server and $\mu_\serv \sim U[1,100] \,\,\forall \serv\in\setS$. }
     %\vspace{-3mm}
    \label{fig:app:evaluation:speed}
\end{figure}

\subsection{Execution run-time}\label{app:Execution run-time}

We repeat the experiments from Section \ref{sec:eval:Execution running times} where we set $\mu_\serv \sim U[1,100]$ instead of $\mu_\serv \sim U[1,10]$. 
The results are depicted in Figure \ref{fig:app:evaluation:speed} and show similar trends. Again, it is evident how the run-time of $SCD$ via Algorithm \ref{alg:opt} scales similarly to $JSQ$ and $SED$ as expected. $SCD$ via Algorithm \ref{alg:n2}, on the other hand, is again slower. 
Interestingly, with the increased heterogeneity, though constant, there is a larger gap among $SCD$, $JSQ$ and $SED$. In particular, $SED$ becomes somewhat slower than $SCD$.

Our investigation revealed the following reason. All algorithm use dedicated and optimized for speed data structures\footnote{for $JSQ$ and $SED$ these are min-heaps that always keep the next best server at the top of the heap so we do not have to sort the server according to their queue length ($JSQ$) or their load ($SED$) after each update.}. However, with the increased heterogeneity, there is an increased gap between the exact number of operations $JSQ$ and $SED$ require to update their data structures when assigning new requests to servers. For $JSQ$, when assigning a new request to a server $\serv$ we need to update $q_\serv(t) \gets q_\serv(t)+1$ disregarding the server processing rate. This results in a predictable behavior of the data structure where only a few operations are needed to fix it (often a single operation). However, for $SED$ the behavior is less predictable. The update in this case is $\frac{q_\serv(t)}{\mu_\serv} \gets \frac{q_\serv(t) + 1}{\mu_\serv}$ where the addition is proportionally inverse the server processing rate, i.e., $\frac{1}{\mu_\serv}$. Thus, the assignment of requests to different servers often results in more required operations to fix the data structure.

\end{document}